\documentclass[a4paper,twocolumn,11pt,accepted=2025-11-04]{quantumarticle}
\pdfoutput=1
\usepackage[utf8]{inputenc}
\usepackage[english]{babel}
\usepackage[T1]{fontenc}
\usepackage{amsmath}
\usepackage{hyperref}

\usepackage{tikz}
\usepackage{lipsum}

\usepackage{graphicx}
\usepackage{dcolumn}
\usepackage{amssymb,amsfonts,amsthm}
\usepackage{mathtools} 
\usepackage{txfonts}
\usepackage{bbm} 
\usepackage{bm} 
\usepackage{mathrsfs} 
\usepackage{xcolor}
\usepackage[numbers,sort&compress]{natbib}
\usepackage{braket}

\allowdisplaybreaks
\usepackage{multicol,blindtext}

\DeclareMathOperator{\tr}{Tr}

\newtheorem{theorem}{Theorem}
\newtheorem{lemma}{\emph{Lemma}}

\newtheorem{corollary}{Corollary}

\newcommand{\norm}[1]{ \lVert #1  \rVert}

\newcommand{\FS}{\mathcal{F}}
\newcommand{\FPT}{\mathcal{E}}

\begin{document}

\title{Stochastic approximate state conversion for entanglement and general quantum resource theories}

\author{Tulja Varun Kondra}
\email{tuljavarun@gmail.com}
\affiliation{Centre for Quantum Optical Technologies, Centre of New Technologies,
University of Warsaw, Banacha 2c, 02-097 Warsaw, Poland}
\affiliation{Institute for Theoretical Physics III, Heinrich Heine University D\"{u}sseldorf, Universit\"{a}tsstra{\ss}e 1, D-40225 D\"{u}sseldorf, Germany}

\author{Chandan Datta}
\affiliation{Centre for Quantum Optical Technologies, Centre of New Technologies,
University of Warsaw, Banacha 2c, 02-097 Warsaw, Poland}
\affiliation{Department of Physics, Indian Institute of Technology Jodhpur, Jodhpur 342030, India}

\author{Alexander Streltsov}
\affiliation{Institute of Fundamental Technological Research, Polish Academy of Sciences, \\ Pawi\'nskiego 5B, 02-106 Warsaw, Poland}

\begin{abstract}
   Quantum resource theories provide a mathematically rigorous way of understanding the nature of various quantum resources. An important problem in any quantum resource theory is to determine how quantum states can be converted into each other within the physical constraints of the theory. The standard approach to this problem is to study approximate or probabilistic transformations. Here, we investigate the intermediate regime, providing limits on both, the fidelity and the probability of state transformations. We derive limitations on the transformations, which are valid in all quantum resource theories, by providing bounds on the maximal transformation fidelity for a given transformation probability. As an application, we show that these bounds imply an upper bound on the asymptotic rates for various classes of states under probabilistic transformations. We also show that the deterministic version of the single copy bounds can be applied for drawing limitations on the manipulation of quantum channels, which goes beyond the previously known bounds of channel manipulations. Furthermore, we completely solve the question of stochastic-approximate state conversion via local operations and classical communication in the following two cases: (i) Both initial and target states are pure bipartite entangled states of arbitrary dimensions. (ii) The target state is a two-qubit entangled state and the initial state is a pure bipartite state.
\end{abstract}

\maketitle

\section{Introduction} The goal of quantum resource theories~\cite{Chitambar_2019,HorodeckiResourceTheories} is to capture the limitations of setups relevant for quantum technologies, thus allowing us to study fundamental phenomena of quantum systems and their usefulness for quantum technological tasks. If one considers two spatially separated parties who can perform all quantum transformations in their local labs and can additionally exchange messages via a classical channel, one arrives at the resource theory of entanglement \cite{Bennett_1996,Vedral_97,Horodecki_2009}. Quantum entanglement is useful for various tasks such as quantum teleportation~\cite{Bennett_1993}, quantum dense coding~\cite{Bennett_1992}, and quantum key distribution~\cite{Ekert_1991}. However, not all quantum technological applications are based on the presence of entanglement, but can rely on other features of quantum systems. For this reason various quantum resource theories have been studied in the recent literature, including quantum thermodynamics~\cite{brandao_TO,Goold_2016}, purity~\cite{HorodeckiPhysRevA.67.062104,GOUR_purity,Streltsov_2018}, coherence~\cite{plennio_coh,WinterPhysRevLett.116.120404,StreltsovPhysRevLett.115.020403,Streltsov_2017}, imaginarity~\cite{Gour2018,Wu_PRL,IM2,varun_imaginarity}, and asymmetry \cite{gour_asymmetry,GourPhysRevA.80.012307}. For continuous-variable systems, quantum resource theories have been studied with a focus on Gaussian states and operations~\cite{Takagi_2018,LamiPhysRevA.98.022335,AlbarelliPhysRevA.98.052350}.

At the heart of every quantum resource theory is the definition of free states and operations, corresponding to quantum states and quantum transformations which are easy to establish or implement. One of the basic questions in every quantum resource theory concerns quantum state conversion: is it possible to transform a quantum state into another one by using a free operation? This question is essential for entanglement theory, as many quantum information processing tasks require singlets for achieving optimal performance. It is thus crucial to develop optimal methods to convert a less useful state into another state which is potentially more useful for the specific task. First results addressing this question concern transformations between pure entangled states. A complete solution to this problem has been given in terms of the Schmidt coefficients of the corresponding quantum states~\cite{Nielsen_1999}. If a conversion is not possible between two states, they still admit a probabilistic or approximate transformation, and results in this direction have been reported for specific setups \cite{vidal_prob,Vidal_approx,thomas_coh,Wu_PRL,IM2,vidal_purestates,JonathanPhysRevLett.83.1455,regula2021probabilistic}.

Many quantum resource theories also allow for asymptotic state conversion, where $n$ copies of a quantum state can be converted into $rn$ copies of another state with a rate $r$, allowing for an error which vanishes in the limit $n \rightarrow \infty$. For pure states, the optimal conversion rate for such asymptotic transformations is known for bipartite entanglement theory~\cite{Bennett_1996} and for the resource theory of coherence~\cite{WinterPhysRevLett.116.120404}. The optimal conversion rates are known for all (pure and mixed) quantum states for the resource theory of purity~\cite{HorodeckiPhysRevA.67.062104}.
The main disadvantage of the asymptotic conversion is the need to manipulate a large number of copies simultaneously. To achieve optimal performance, this procedure typically requires entangled measurements across many particles, which are out of reach for current quantum technologies. 

In this article, we first focus on state transformations which involve a single copy of a quantum state. We consider the intermediate regime between probabilistic and approximate transformations. Here, the goal is to convert a quantum state $\rho$ into another quantum state $\sigma$ with maximal probability, allowing for a small error in the conversion. 

This intermediate regime between probabilistic and approximate transformations has been very relevant in many practical schemes of entanglement manipulation, as they explicitly allow for a small probability of failure and the optimal fidelity is considered in the case of success~\cite{Desef2022optimizingquantum,PhysRevA.97.062333,PhysRevX.4.041041,PhysRevLett.76.722,PhysRevLett.77.2818,PhysRevA.64.014301,PhysRevA.64.012304,WOS:000168285500040,PhysRevLett.101.130502}. Specifically in quantum networks, such a trade-off between the probability of success and the achievable fidelity is highly relevant~\cite{WOS:000172304500034,PhysRevA.71.060310}. Therefore, it is natural to develop a general framework for the manipulation of quantum resources within such probabilistic–approximate settings.

This intermediate regime has previously been studied in \cite{Fang_2019,Regula_distillation,regula2021probabilistic,9834375,Regula2022tightconstraints}. All previous works made the standard assumption that the target state (or channel) is pure \cite{Fang_2019,Regula_distillation,regula2021probabilistic,9834375,Regula2022tightconstraints}. However, this assumption often fails in many resource theories, where the golden resource is not pure. Notable examples include resource theories of distinguishability \cite{PhysRevResearch.1.033170,Salzmann_2021}, Bell nonlocality \cite{RevModPhys.86.419}, and measurement incompatibility\cite{Heinosaari_2016}. We derive restrictions on the achievable fidelity and probability at the single copy level, which go beyond this purity assumption. We show that the deterministic version of the single copy bounds can be utilized to set restrictions on transformation of channels.

In addition, our bounds are also applicable to resource theories where the free set is non-convex, such non-convex resource theories have been of interest in recent times \cite{Kuroiwa_2024,Kuroiwa_2024_1,salazar2024convexity}. Interestingly, our bounds depend solely on the resource quantifiers of the initial and target states, a simplicity that enables their direct application to many-copy scenarios. We further analyze their behavior in the asymptotic limit, where they provide upper bounds on the achievable transformation rates.

We then turn to a concrete setting: for bipartite entangled state transformations, we provide a complete analytical solution for pure states of arbitrary dimension, as well as for two-qubit systems when the initial state is pure.
\medskip

\section{Stochastic approximate state conversion}
An important ingredient of every quantum resource theory is the set of free states $\FS$. Typically, this set corresponds to all states which can be easily prepared. In entanglement theory the free states are separable states \cite{Werner_1989,Horodecki_2009}. Every state which is not free is called resource state. Another important ingredient of any resource theory is the set of free operations. Free operations satisfy a necessary condition that they do not create resource states from free states, i.e., if $\Lambda$ is a free operation, then
\begin{equation}
    \Lambda[\rho_f] \in \FS \label{eq:LambdaF}
\end{equation}
for any free state $\rho_f$.

In general, not every operation which follows Eq.~(\ref{eq:LambdaF}) is a free operation. An example is the theory of entanglement, where some quantum operations do not create entanglement, but cannot be implemented with the means of Local Operations and Classical Communication (LOCC)~\cite{BennettPhysRevA.59.1070}. A more axiomatic approach to quantum resource theories is to make Eq.~(\ref{eq:LambdaF}) the only requirement for a free operation. The set of operations defined in this way is the maximal possible set within any meaningful resource theory, as any operation outside of this set will necessarily create a resource state from some free state. In the same spirit, one can also define a probabilistic free transformation ($\FPT$), such that $\FPT[\rho_f]/\tr(\FPT[\rho_f]) \in \FS$ for any free state $\rho_f$. Here, $\FPT[\rho] = \sum_j K_j \rho K_j^\dagger$ can be represented by 
 a (possibly incomplete) set of Kraus operators $\{K_j\}$ and the transformation probability is given by $p = \tr(\FPT[\rho])$. 

Given two quantum states $\rho$ and $\sigma$, deterministic free conversion $\rho \rightarrow \sigma$ is possible if $\sigma = \Lambda_{\mathrm f}[\rho]$ for some free operation $\Lambda_{\mathrm f}$. For every nontrivial resource theory there exist some pairs of states for which no deterministic free conversion is possible. In such cases, there is still a chance for probabilistic conversion, where the state $\rho$ is converted into $\sigma$ with some nonzero probability. The maximal probability for such probabilistic conversion is then defined as~\cite{IM2}
\begin{align}\label{eq:probablistic_exact}
P\left(\rho\rightarrow\sigma\right) & =\max_\FPT \left\{ \tr\left(\FPT\left[\rho\right]\right):\frac{\FPT\left[\rho\right]}{\tr\left(\FPT\left[\rho\right]\right)}=\sigma\right\},
\end{align}
and the maximum is taken over all free probabilistic transformations $\FPT$. While evaluating the optimal conversion probability is challenging in general, closed expressions have been obtained for special cases in various quantum resource theories, e.g. entanglement \cite{vidal_purestates,vidal_prob}, coherence \cite{PhysRevA.101.012313,coherence_pure,Erratum_coh_pure,thomas_coh}, and imaginarity \cite{Wu_PRL,IM2}. 
A powerful upper bound on the conversion probability $P$ is given by $P(\rho\rightarrow\sigma)\leq\min\left\{ \mathcal{R}(\rho)/\mathcal{R}(\sigma),1\right\}$, where $\mathcal R$ is any strongly monotonic resource quantifier~\cite{Vidaldoi:10.1080/09500340008244048,IM2}.


In most quantum resource theories there exist pairs of states which do not allow neither for deterministic nor probabilistic conversion. In such a case, it is always possible to achieve an approximate transformation between $\rho$ and $\sigma$ with conversion fidelity
\begin{equation}\label{eq:fidelity}
F(\rho\rightarrow\sigma)=\max_{\Lambda_{\mathrm f}}F(\Lambda_{\mathrm f}[\rho],\sigma),
\end{equation} 
where the maximum is taken over all free operations $\Lambda_{\mathrm f}$ and $F(\rho,\sigma)=[\tr(\sqrt{\rho}\sigma\sqrt{\rho})^{1/2}]^{2}$. Closed expressions for the optimal achievable fidelity have been noted for various resource theories, for examples, entanglement \cite{Vidal_approx}, coherence \cite{PhysRevResearch.4.023199}, imaginarity \cite{Wu_PRL,IM2}. 
Results concerning approximate transformations in general resource theories have also been presented in~\cite{RegulaPhysRevA.101.062315}.


Probabilistic and approximate conversion are special cases of a more general transformation: while probabilistic conversion is assumed to produce a state with unit fidelity, approximate conversion is producing a state with unit probability. Very little is known about the intermediate regime between probabilistic and approximate transformations \cite{regula2021probabilistic,Fang_2019,Regula_distillation}. For investigating this intermediate regime, we define the \emph{fidelity for stochastic approximate state conversion}. It quantifies the maximal fidelity for the transformation from $\rho$ to $\sigma$ with a conversion probability at least $p$:
\begin{eqnarray}
F_{p}(\rho\rightarrow\sigma)=\max_{\FPT}\bigg\{ F\left(\frac{\FPT[\rho]}{\tr(\FPT[\rho])},\sigma\right):\nonumber\\
\tr\left(\FPT[\rho]\right)\geq p\bigg\}. \label{eq:FidelityPSC}
\end{eqnarray}
Here, the maximum is taken over all free probabilistic transformations $\FPT$. In the same spirit, we define the \emph{probability for stochastic approximate state conversion}, capturing the maximal conversion probability for a conversion with fidelity at least $f$:
\begin{eqnarray}
P_{f}(\rho\rightarrow\sigma)=\max_{\FPT}\bigg\{ &&\tr\left(\FPT[\rho]\right):\nonumber\\
&& F\left(\frac{\FPT[\rho]}{\tr\left(\FPT[\rho]\right)},\sigma\right)\geq f\bigg\}. \label{eq:ProbabilityPSC}
\end{eqnarray}

 
 In this work, we consider a more general set of operations ($\Lambda_\varepsilon$) which generate at most $\varepsilon$ amount of resource (according to some resource measure $M$) i.e., $M\left(\frac{\Lambda_\varepsilon[\rho_f]}{\tr \Lambda_\varepsilon[\rho_f]}\right)\leq \varepsilon$, for all free states $
 \rho_f$. We call such operations as $\varepsilon$-resource generating operations. Once again, for such operations, one can define the above quantities in Eqs. (\ref{eq:FidelityPSC}) and (\ref{eq:ProbabilityPSC}) as
\begin{eqnarray}
F^{M, \varepsilon}_{p}(\rho\rightarrow\sigma)=\max_{\Lambda_\varepsilon}\bigg\{ &&F\left(\frac{\Lambda_\varepsilon[\rho]}{\tr(\Lambda_\varepsilon[\rho])},\sigma\right):\nonumber\\
&&\tr\left(\Lambda_\varepsilon[\rho]\right)\geq p\bigg\}. \label{eq:FidelityPSC_epsilon}
\end{eqnarray}
and
\begin{eqnarray}
P^{M, \varepsilon}_{f}(\rho\rightarrow\sigma)=&&\max_{\Lambda_\varepsilon}\bigg\{ \tr\left(\Lambda_\varepsilon[\rho]\right):\nonumber\\
&&F\left(\frac{\Lambda_\varepsilon[\rho]}{\tr\left(\Lambda_\varepsilon[\rho]\right)},\sigma\right)\geq f\bigg\}
\end{eqnarray}
respectively. Note that, here the maximisation is taken over all $\varepsilon$-resource generating operations $\Lambda_{\varepsilon}$.

We will now introduce three resource measures, the geometric measure $G$, the generalised robustness $R$ and the standard robustness $S$, which will be used later:
\begin{align}
G(\rho) & =1-\max_{\sigma\in\FS}F(\rho,\sigma)\\
R(\rho) & =\min_{\tau}\left\{ s\geq0:\frac{\rho+s\tau}{1+s}\in\FS\right\}\,\,\text{and}\\
S(\rho) & =\min_{\tau\in\FS}\left\{ s\geq0:\frac{\rho+s\tau_f}{1+s}\in\FS\right\}.
\end{align}
All the above quantities are non-negative, vanish for free states, and are non-increasing under free operations \cite{Chitambar_2019}. The quantifiers $G$, $R$ and $S$ have been initially introduced in the context of entanglement \cite{Vidal_robustness,SteinerPhysRevA.67.054305,HarrowPhysRevA.68.012308,Shimony1995,Barnum2001,Wei_2003,Alex_Linking}, and recently also found applications in general quantum resource theories~\cite{TakagiPhysRevLett.122.140402,TakagiPhysRevX.9.031053}, including the resource theories of coherence~\cite{NapoliPhysRevLett.116.150502,StreltsovPhysRevLett.115.020403} and imaginarity~\cite{Wu_PRL,IM2,Gour2018}. Both the generalised robustness ($R$) and geometric measure ($G$) can be alternatively written in the following way \cite{datta_robust, wilde_book}
\begin{eqnarray}
&& E_{1/2}(\rho) =-\log_2(1-G(\rho))=-\log_2F_{\max}(\rho),\nonumber\\
\\
&& E_{\max}(\rho)  =\log_2(1+R(\rho))
\end{eqnarray}
where,
\begin{eqnarray}
    && E_{1/2}(\rho) =\min_{\sigma\in\FS} D_{1/2}(\rho||\sigma),\\
    && E_{\max}(\rho) = \min_{\sigma\in\FS}D_{\max}(\rho||\sigma).
\end{eqnarray}
Here, $D_{\max}(\rho||\sigma)=\lim\limits_{\alpha\rightarrow\infty}D_{\alpha}(\rho||\sigma)$, where $D_{\alpha}$ is the sandwiched R\'{e}nyi relative entropy. It is known that, for all $\rho$ and $\sigma$, $D_{\alpha}(\rho||\sigma)\leq D_{\alpha'}(\rho||\sigma)$ for all $\alpha\leq\alpha'$ and $D_{\alpha}(\rho||\sigma)\leq D_{\alpha}(\Lambda(\rho)||\Lambda(\sigma))$ for every CPTP map $\Lambda$ and $\alpha\in[\frac{1}{2},\infty)$ (see section 7.5 of \cite{wilde_book} for extensive discussion on sandwiched R\'{e}nyi entropies). Therefore,
\begin{equation}
     D_{1/2}(\rho||\sigma)\leq D_{\max}(\rho||\sigma).
\end{equation}
and
\begin{eqnarray}
    \min_{\sigma\in\FS} D_{1/2}(\rho||\sigma)\leq D_{1/2}(\rho||\sigma')&&\leq D_{\max}(\rho||\sigma')\\
    &&=\min_{\sigma\in\FS}D_{\max}(\rho||\sigma).\nonumber
\end{eqnarray}
This shows that for any state $\rho$, 
\begin{eqnarray}
     E_{1/2}(\rho)\leq E_{\max}(\rho).
\end{eqnarray}
This can be equivalently written as,
\begin{equation}
R(\rho) \geq \frac{G(\rho)}{1-G(\rho)}. \label{r,g_bound}
\end{equation}
A similar inequality between the resource robustness and the geometric measure (when the set of free states is obtained as the convex hull of free pure states) has been reported in \cite{Regula_Linking,Cavalcanti_2006}: $R(\rho)\geq\tr(\rho^{2})/[1-G(\rho)]-1$. It is straightforward to see that this bound coincides with Eq.~(\ref{r,g_bound}) for all pure states. Moreover, Eq.~(\ref{r,g_bound}) gives an improvement whenever $\rho$ is not pure and $G(\rho) > 0$. An analogous bound to Eq.~(\ref{r,g_bound}) also holds for quantum channels (see Appendix \ref{app_channel_bound}). 

Equipped with these tools we are now ready to prove the following theorem.
\begin{theorem}\label{thm:GeneralBound_resource_generating}
    For any quantum resource theory and any two states $\rho$ and $\sigma$ the following inequalities hold:
 \begin{align}
F^{M, \varepsilon}_{p}\left(\rho\rightarrow\sigma\right) & \leq\min\left\{\frac{1}{p}\left[1+R(\rho)\right]\times F^{M, \varepsilon}_{\max}(\sigma),1\right\},\label{eq:Main3}\\
P^{M, \varepsilon}_{f}\left(\rho\rightarrow\sigma\right) & \leq\min\left\{\frac{1}{f}\left[1+R(\rho)\right]\times F^{M, \varepsilon}_{\max}(\sigma),1\right\}.\label{eq:Main4}
\end{align}
Here, $F^{M, \varepsilon}_{\max}(\sigma)=\max\limits_{\sigma_{\varepsilon}: M(\sigma_{\varepsilon})\leq \varepsilon}F(\sigma,\sigma_{\varepsilon})$.
\end{theorem}
The proof of the above theorem can be found in Appendix \ref{app_theorem1}. Theorem~\ref{thm:GeneralBound_resource_generating} provides general upper bounds for the fidelity and probability for stochastic approximate state conversion, which are valid for any quantum resource theory. Going one step further, we can also allow for an error in the intial state $\rho$. From now on, we always consider $M$ to be either generalised robustness ($R$) or standard robustness ($S$). Robustness (generalised) generating operations have been previously studied, in the context of performing reversible manipulation of quantum resources \cite{Berta2023gapinproofof,PhysRevLett.115.070503,regula2024reversibility}. Note that, when $M$ is either standard or generalised robustness, the following holds
\begin{equation}\label{simple_rob}
    F^{M, \varepsilon}_{\max}(\sigma)\leq F_{\max}(\sigma)(1+\varepsilon).
\end{equation}
The proof for the above statement is in Appendix \ref{app_asymptotic}. For the special case, when $\varepsilon=0$ i.e., for resource non-generating operations, we arrive at the following result.
\begin{corollary} \label{thm:GeneralBound}
For any quantum resource theory and any two states $\rho$ and $\sigma$ the following inequalities hold:
\begin{eqnarray}
    F_{p}\left(\rho\rightarrow\sigma\right) && \leq\min\left\{\frac{1}{p}\left[ 2^{E_{\max}(\rho) -E_{1/2}(\sigma)} \right],1\right\}\label{eq:Main1} \\
    P_{f}\left(\rho\rightarrow\sigma\right) && \leq \min\left\{\frac{1}{f}\left[ 2^{E_{\max}(\rho) -E_{1/2}(\sigma)} \right] ,1\right\}.\label{eq:Main2}
\end{eqnarray}
\end{corollary}
Cases of special interest are deterministic conversion with $p=1$ and exact conversion with $f=1$. In these cases Eqs.~(\ref{eq:Main1}) and (\ref{eq:Main2}) reduce to 
\begin{align}
F\left(\rho\rightarrow\sigma\right) & \leq\min\left\{ 2^{E_{\max}(\rho) -E_{1/2}(\sigma)} ,1\right\}, \label{eq:GeneralBoundDeterministic}\\
P\left(\rho\rightarrow\sigma\right) & \leq\min\left\{ 2^{E_{\max}(\rho) -E_{1/2}(\sigma)} ,1\right\}.
\end{align}

\noindent Note that the above bounds are non-trivial whenever $E_{\max}(\rho) < E_{1/2}(\sigma)$ and they only depend on the free states of the theory and are independent from the particular choice of free operations. The bound given in Eq. (\ref{eq:Main1}) has been previously reported for pure target states \cite{Regula_distillation} and shown to outperform the best known bounds in certain scenarios. \

Along the same line, we can also provide a deterministic bound for transformation of channels, as below.
\begin{eqnarray}\label{channelbound_main1}
F^{M, \varepsilon}\left(\Theta(\mathcal{E}),\mathcal{E}_t\right)\leq \min\Big\{ &&(1+R(\mathcal{E}))\times \nonumber\\
&& F^{M_{\varepsilon}}_{\max}(\mathcal{E}_t) ,1\Big\}.
\end{eqnarray}
Here, $\mathcal{E}$ is the initial channel, $\mathcal{E}_t$ is the target channel and $\Theta$ is a $\varepsilon$-resource generating super-channel and therefore for the case of resource non-generating super-channels ($\varepsilon=0$), we have
\begin{eqnarray}\label{channelbound_main2}
F\left(\Theta_f(\mathcal{E}),\mathcal{E}_t\right)\leq \min\Big\{ &&(1+R(\mathcal{E}))\times\nonumber\\&&( 1-G(\mathcal{E}_t)) ,1\Big\}.
\end{eqnarray}
The proof, along with all the relevant definitions can be found in Appendix \ref{app_channel_bound}.

Note that both the quantities appearing in the above bounds i.e, $E_{\max}$ (generalised robustness) and $E_{1/2}$ (geometric measure) can be computed as convex optimization problems whenever the set of free states (or channels) is convex (this is true for most of the known resource theories). Additionally, when the free set is SDP representable, as in the case for state based resource theories like imaginarity, coherence, asymmetry \cite{PhysRevA.93.042107,Wu_PRL} and channel based resource theories like measurement incompatibility, nonlocality \cite{https://doi.org/10.48550/arxiv.2205.08546,PhysRevLett.125.210402}, both these quantities can be computed by a semidefinite program (SDP).

 In many cases these bounds may not be tight and better bounds can be found for specific setups. For example, distillation of two-qubit entanglement by LOCC, where the optimal fidelity of achieving a maximally entangled two-qubit state ($\ket{\phi^+} = (\ket{00}+\ket{11})/\sqrt 2$) for a given probability $p$ is given by 
 \begin{equation}\label{eq:distill_singlet_locc}
      F_{p}\left(\rho\rightarrow\phi_{+}\right)\leq\min\left\{\frac{R(\rho)}{2p}+\frac{1}{2},1\right\}.
 \end{equation}
  The details can be found in Appendix \ref{app_distill_locc}. Note that for $p<1$, the bound is always tighter than the bound given in Eq. (\ref{eq:Main1}). For the case when target state is a maximally entangled state of arbitrary dimension, one can provide a tighter (compared to Eq. (\ref{eq:Main2})) SDP-computable upperbound on the achievable fidelity of transformation via LOCC, given a lower bound on the probability of success \cite{PhysRevA.97.062333}.

  The channel bounds in Eqs. (\ref{channelbound_main1}) and (\ref{channelbound_main2}) go beyond the bounds presented in \cite{Regula_distillation,RegulaPhysRevA.101.062315}. This is because, the previous bounds presented for channel transformations, assume that the target channel is ``pure'' (pure channels map all pure sates into pure states). In many channel based resource theories the ``standard target channels'' are not pure \cite{ryuji_oneshot}. For example, in the case of resource theory of measurement incompatibility, mutually unbiased basis measurements correspond to the standard target resource and in the case of Bell nonlocality, the standard target channels are the PR boxes, very much like singlet in the resource theory of entanglement \cite{ryuji_oneshot,PhysRevLett.95.140401}. Both the previously mentioned standard target channels are not pure. Therefore, the previously known bounds from \cite{Regula_distillation,RegulaPhysRevA.101.062315} do not apply to these cases. However, the channel bound from Eq. (\ref{channelbound_main2}) can be used in these scenarios, as we do not put any restriction of purity on the target channel. As an example, we compute the bound for distillation of a PR box ($\mathcal{B}_{\textrm{pr}}$) correlations starting from an isotropic box $\mathcal{B}_{\eta}=\eta\mathcal{B}_{\textrm{pr}}+(1-\eta)\mathcal{B}_{\textrm{random}}$, where $\eta\in[0,1]$ and $\mathcal{B}_{\textrm{random}}$ represents a uniformly distributed box. For a general scenario, computation of the robustness of Bell-nonlocality ($R_{\textrm{nl}}$) for a given box, is an SDP \cite{PhysRevLett.125.210402}. For $\mathcal{B}_{\eta}$, the analytic expression for robustness is given by $R_{\textrm{nl}}(B_{\eta})=\frac{2\eta -1}{3}$ for $\eta \geq \frac{1}{2}$ and $R_{\textrm{nl}}(B_{\eta})= 0$ for $\eta<\frac{1}{2}$\cite{Wolfe_2020}. In Appendix \ref{app_prbox_ditill}, we show that the closest box (largest fidelity) to $\mathcal{B}_{\textrm{pr}}$ in the set of local boxes is $\mathcal{B}_{\frac{1}{2}}=(\mathcal{B}_{\textrm{pr}}+\mathcal{B}_{\textrm{random}})/2$ (note that $\mathcal{B}_{\frac{1}{2}} \equiv \mathcal{B}_{\eta}\,\,\text{when}\,\,\eta=\frac{1}{2}$) and the fidelity between them is $3/4$. Equivalently, the geometric measure of the PR box is $1/4$. Therefore, the upper bound (see Eq. (\ref{channelbound_main2})) on the approximate distillation of a PR box from an isotropic box is given by
\begin{equation}\label{PR_bound}
    F(\mathcal{B}_{\eta}\rightarrow\mathcal{B}_{\textrm{pr}})\leq \frac{\eta +1}{2}\,\,\text{for}\,\,\frac{1}{2}\leq\eta\leq 1.
\end{equation}
When $\eta\leq1/2$, $\mathcal{B}_{\eta}$ is a local box and the optimal fidelity is achieved by going to  $\mathcal{B}_{\frac{1}{2}}$, which is $3/4$. Note that, the fidelity between $\mathcal{B}_{\eta}$ and $\mathcal{B}_{\textrm{pr}}$ is $\frac{\eta +1}{2}$. Therefore, the bound in Eq. (\ref{PR_bound}) is tight and shows that locality preserving superchannels do not increase the fidelity between isotropic box and PR box (when $1/2\leq\eta\leq 1$). This agrees with the fact that, nonlocality distillation is impossible from an isotropic box. See Appendix \ref{app_prbox_ditill} for the detailed calculation.

In the following, we discuss another example with a class of correlated non-local boxes which are expressed as
\begin{equation}
    \mathcal{B}^{\textrm{c}}_{\eta} = \eta \mathcal{B}_{\textrm{pr}}+(1-\eta)\mathcal{B}^{\textrm{c}},
\end{equation}
where, 
\begin{equation}
    \mathcal{B}^{c}(ab|xy)=
    \begin{cases}
    \frac{1}{2} \,\, \mbox{for}\,\, a\oplus b=0 \\
    0 \,\, \mbox{otherwise}
    \end{cases}
\end{equation}
is a correlated local box. For such non-local boxes ($ \mathcal{B}^{\textrm{c}}_{\eta}$), the value of CHSH violation is given by $2(\eta +1)$. Therefore, from \cite{Wolfe_2020}, it follows that, $R(\mathcal{B}^{\textrm{c}}_{\eta})=\frac{2\eta}{2\eta +4}$. Therefore, the upper bound (see Eq. (\ref{channelbound_main2})) on approximate distillation of a PR box from an correlated box is given by
\begin{equation}\label{PR_bound}
    F(\mathcal{B}^{\textrm{c}}_{\eta}\rightarrow\mathcal{B}_{\textrm{pr}})\leq \frac{3(\eta+1)}{2\eta+4}\,\,\text{for}\,\,0\leq\eta\leq 1.
\end{equation}
Whereas, $F(\mathcal{B}^{\textrm{c}}_{\eta}$, $\mathcal{B}_{\textrm{pr}})=\eta<\frac{3(\eta+1)}{2\eta+4}$ for all $\eta\in[0,1)$. This gives an upper bound to the achievable fidelity between a correlated non-local box and a PR box.

  \medskip
\section{Asymptotic and many-copy bounds} We now study the nature of our bounds in the asymptotic limit. As we show now, our single-copy bounds imply upper bounds on the asymptotic transformation rates in general resource theories. The deterministic rate for a transformation between $\rho$ and $\sigma$ is given by
\begin{eqnarray}\label{det_asymp}
     r(\rho\rightarrow &&\sigma) =
\sup\Big\{r: \nonumber\\
    && \lim_{n\rightarrow\infty}\inf_{\Lambda_f}\norm{\Lambda_{f}(\rho^{\otimes n})-\sigma^{\otimes \lfloor rn\rfloor}}_{1}=0\Big\}.
\end{eqnarray}
Here, the infimum is taken over the set of deterministic free operations ($\Lambda_f$). We now generalise the above definition to the probabilistic case where the probability of success is not allowed to decay too fast (exponentially in the number of copies). We also allow for generation of small amounts of resource, quantified by  $M$. Here, small amounts mean sub-exponential in the number of copies of the state. In such a scenario, we can define asymptotic rates as follows
\begin{eqnarray}\label{gen_asymp}
    &&r^M_{p}(\rho\rightarrow\sigma)=\sup\bigg\{r:\nonumber\\
    && \lim_{n\rightarrow\infty}\inf_{\Lambda_{\varepsilon_{n}}}\norm{\frac{\Lambda_{\varepsilon_{n}}(\rho^{\otimes n})}{\tr \Lambda_{\varepsilon_{n}}(\rho^{\otimes n})}-\sigma^{\otimes \lfloor rn\rfloor}}_{1}=0\bigg\}
\end{eqnarray}
Here, $\Lambda_{\varepsilon_{n}}$ are $\varepsilon_n$-resource (quantified by $M$) generating operations, $\lim\limits_{n\rightarrow\infty}\frac{\log\varepsilon_n}{n}= 0$ and $\lim\limits_{n\rightarrow\infty}-\frac{\log \tr \Lambda_{\varepsilon_{n}}(\rho^{\otimes n})}{n}=0$. Once we have the above definitions, we now present our theorem, which gives an upper bound on the achievable rate.
\begin{theorem}\label{thm:asymptotic}
For any quantum resource theory and any two states $\rho$ and $\sigma$, such that $E_{1/2}(\sigma^{\otimes n})=n\cdot E_{1/2}(\sigma)$, it holds that
\begin{equation}
    r^{M}_{p}(\rho\rightarrow\sigma)\leq \frac{E_{\max}(\rho)}{E_{1/2}(\sigma)},
    \end{equation}
    where $M$ can be $R$ (generalised robustness) or $S$ (standard robustness).
\end{theorem}
The proof of this theorem can be found in Appendix \ref{app_asymptotic}. Additionally, we also show that (in Appendix \ref{app_asymptotic}), whenever we try to achieve a rate $r > r^{M}_{p}(\rho\rightarrow\sigma)=\frac{E_{\max}(\rho)}{E_{1/2}(\sigma)}$, the error of transformation exponentially increases with $n$. Note that the assumption $E_{1/2}(\sigma^{\otimes n})=n\cdot E_{1/2}(\sigma)$ (additivity of $E_{1/2}$) holds for various classes of states
in several resource theories. For example, in the resource theory of coherence \cite{Zhu_2017}, in the resource theory of magic, for single qubit states and for pure states \cite{rubboli2023mixedstate} and in the resource theory of entanglement, for
pure, maximally correlated, GHZ, Bell diagonal, isotropic, and generalized Dicke states \cite{https://doi.org/10.48550/arxiv.2211.12804}.
Recently in Ref. \cite{PhysRevA.107.042401}, the authors show that probabilistic protocols of entanglement distillation can asymptotically transform certain bipartite quantum states (lets say $\rho$) into $\ket{\phi^+} = (\ket{00}+\ket{11})/\sqrt 2$, with a rate larger than $E_{\max}(\rho)$ (note that $E_{1/2}(\ket{\phi^+}\bra{\phi^+})=1$). Using Theorem \ref{thm:asymptotic}, one can say that any such protocol would require the success probability to vanish exponentially in the number of copies of $\rho$. Thus providing an immediate application of Theorem \ref{thm:asymptotic}. 

Alternatively, one can consider a scenario when one aims to transform $n$ copies of the initial $\rho$ into $m$ copies of the target state $\sigma$ via $\varepsilon$-resource generating operations (resource measured as standard or generalised robustness). In such a case, the bounds in Eq. (\ref{eq:Main3}) and Eq. (\ref{eq:Main4}), along with Eq. (\ref{simple_rob})  gives 
\begin{eqnarray}
    &&F^{M,\varepsilon}_{p}\left(\rho^{\otimes n}\rightarrow\sigma^{\otimes m}\right)\nonumber\\ && \leq \min\left\{\frac{1+\varepsilon}{p}\left[ 2^{E_{\max}(\rho^{\otimes n}) -E_{1/2}(\sigma^{\otimes m})}\right] ,1\right\}, \\
    &&P^{M,\varepsilon}_{f}\left(\rho^{\otimes n}\rightarrow\sigma^{\otimes m}\right)\nonumber\\ && \leq \min\left\{\frac{1+\varepsilon}{f} \left[2^{E_{\max}(\rho^{\otimes n}) -E_{1/2}(\sigma^{\otimes m})}\right] ,1\right\}.
\end{eqnarray}
In order to evaluate the above given bounds, one has to evaluate the resource measures on multiple copies of the intial and target states. But thanks to the sub-additivity of $E_{\max}$ and (assumed) additivity of $E_{1/2}$, one can provide alternative bounds, such that the resource measures ($E_{\max}$ and $E_{1/2}$) need to be evaluated only on a single copy of initial and target state respectively. 
\begin{eqnarray}
    && F^{M,\varepsilon}_{p}\left(\rho^{\otimes n}\rightarrow \sigma^{\otimes m}\right) \nonumber\\
    &&\leq \min\left\{\frac{1+\varepsilon}{p}\left[ 2^{n E_{\max}(\rho) - m E_{1/2}(\sigma)}\right] ,1\right\}, \\
    &&P^{M,\varepsilon}_{f}\left(\rho^{\otimes n}\rightarrow\sigma^{\otimes m}\right)\nonumber\\ && \leq \min\left\{\frac{1+\varepsilon}{f}\left[ 2^{n E_{\max}(\rho) - m E_{1/2}(\sigma)} \right],1\right\}.
\end{eqnarray}

\section{Pure entangled state transformations} 
The definition of the fidelity and probability for probabilistic approximate state conversion in Eqs.~(\ref{eq:FidelityPSC}) and (\ref{eq:ProbabilityPSC}) suggests that an analytic expression for $F_p$ and $P_f$ is out of reach in most setups. As we will see below, an analytic expression can indeed be found for relevant setups within the resource theory of entanglement. Previously, bipartite state transformations for probabilistic exact (see Eq. (\ref{eq:probablistic_exact})) \cite{vidal_purestates} and deterministic approximate (see Eq. (\ref{eq:fidelity})) \cite{Vidal_approx} scenarios have been explored. However, one can consider a more general scenario, where there is an interplay between achievable probability and fidelity. Here, in Theorem~\ref{thm:pureQuditConversion}, we explore this scenario and provide a complete solution for probabilistic approximate bipartite pure state transformations. Let us define $n=\max\{\text{Sch}(\psi),\text{Sch}(\phi)\}$. Note that $\{\alpha_i\}$ and $\{\beta_i\}$ (squared Schmidt coefficients in decreasing order) can both be considered to be $n$-dimensional column vectors. This can be done by adding $n-\text{Sch}(\psi)$ zeros to $\{\alpha_i\}$ if $\text{Sch}(\psi)<n$ and when $\text{Sch}(\phi)<n$, one adds $n-\text{Sch}(\phi)$ zeros to $\{\beta_i\}$. For the sake of convenience, from now onwards we drop the labels of the systems of Alice and Bob. As an example, we shall use $\ket{\psi}$ instead of $\ket{\psi^{AB}}$.

\begin{theorem}\label{thm:pureQuditConversion}
The optimal fidelity to convert a bipartite pure state $\ket{\psi}$ into another state $\ket{\phi}$ via LOCC with a probability $p$ is given by
\begin{equation}\label{Fidelity_arbitrary}
    F_p(\psi\rightarrow\phi)=\min_{l\in\{2,\cdots,n\}}\left\{1-4\left( E_{l}^{\phi}-\frac{ E_{l}^{\psi}}{p}\right)^2\right\},
\end{equation}
where $E^{\psi}_l=\sum_{i=l}^{n}\alpha_i$,
   $E^{\phi}_l=\sum_{i=l}^{n}\beta_i$ and $\{\alpha_i\}$, $\{\beta_i\}$ are the squared Schmidt coefficients of $\psi$ and $\phi$ in decreasing order. 
\end{theorem}

\noindent For the proof see Appendix \ref{app_proof_pure_bipartite}. In the above scenario, a closed expression for the maximum conversion probability $P_f$ for a given fidelity can also be found: 
\begin{eqnarray}\label{Probability_arbitrary}
    &&P_f(\psi\rightarrow\phi)=\nonumber\\
    &&\min\left\{1,\min_l \left\{\frac{E_l^{\psi}}{\max\left\{E_l^{\phi}- \frac{\sqrt{1-f}}{2},0\right\}}\right\} \right\},
\end{eqnarray}
where $l$ ranges from 2 to $n$. Details can be found in Appendix \ref{app_proof_pure_bipartite}. Therefore, the results in Eq. (\ref{Fidelity_arbitrary}) and (\ref{Probability_arbitrary}) provide a complete solution for the bipartite pure entangled state transformations. Note that by putting $p=1$ in Eq. (\ref{Fidelity_arbitrary}) and $f=1$ in Eq. (\ref{Probability_arbitrary}), we get back the known result for deterministic approximate transformations given in \cite{Vidal_approx} and probabilistic exact transformations given in \cite{vidal_purestates} respectively. 

\medskip

\section{Entangled states of two qubits}In the following, we consider another scenario, where the initial state is pure bipartite and the target state is an arbitrary two-qubit state. Here also exact expressions of $F_p$ and $P_f$ can be found. Before presenting the results, we provide tight continuity bounds on geometric entanglement in two-qubit settings, which will be utilised in the derivation of the analytical expressions for $F_p$ and $P_f$.

\begin{lemma} \label{thm:TwoQubitsContinuity}
 For a two-qubit state $\rho$ consider a set of states $S_{\rho,f}$ such that $F(\rho,\rho')\geq f$ for all $\rho' \in S_{\rho,f}$. The minimal geometric entanglement in $S_{\rho,f}$ is given by 
 \begin{eqnarray}
 &&\min_{\rho'\in S_{\rho,f}}G(\rho')= \nonumber\\ &&\sin^{2}\left(\max\left\{ \sin^{-1}\!\sqrt{G(\rho)}-\cos^{-1}\!\sqrt{f},0\right\} \right).\nonumber\\
 \end{eqnarray}
 For pure two-qubit states $\ket{\psi}$, the maximal geometric entanglement in $S_{\psi,f}$ is given by 
 \begin{eqnarray}
     &&\max_{\rho'\in S_{\psi,f}}G(\rho')  =\nonumber\\
     &&\sin^{2}\left(\min\left\{ \sin^{-1}\!\sqrt{G(\psi)}+\cos^{-1}\!\sqrt{f},\frac{\pi}{4}\right\} \right),\nonumber\\
 \end{eqnarray}
 where $\psi=\ket{\psi}\!\bra{\psi}$ denotes a projector onto a pure states $\ket{\psi}$.
 \end{lemma}

 See Appendix \ref{app_tight_continuinty} for the proof. These bounds shed new insight into the features of two-qubit entanglement. We are now ready to give a complete solution for the stochastic approximate state conversion in two-qubit systems if the initial state is pure.

\begin{theorem} \label{thm:QubitConversion}
The maximal probability to convert a pure two-qubit state $\ket{\psi}$ into a two-qubit state $\rho$ via LOCC with a fidelity $f$ is given by
\begin{eqnarray}\label{eq:optimal_probability_entanglement}
&&P_{f}(\psi\rightarrow\rho)\nonumber\\ &&= \begin{cases}
1\,\,\,\mathrm{for}\,\,\,m_1\geq0\\
\frac{G(\psi)}{\sin^{2}\left(\sin^{-1}\sqrt{G(\rho)}-\cos^{-1}\sqrt{f}\right) }\,\,\,\mathrm{otherwise},
\end{cases}
\end{eqnarray}
where $m_1=\sin^{-1}\sqrt{G(\psi)}-\sin^{-1}\sqrt{G(\rho)} +\cos^{-1}\sqrt{f}$.
\end{theorem}
\noindent We refer to Appendix  \ref{app_fid_prob_two_qubit} for more details. While Theorem \ref{thm:QubitConversion} gives a closed expression for the maximal conversion probability $P_f$, it is also possible to obtain a closed expression for the maximal fidelity $F_p$: 
\begin{eqnarray}
&&F_{p}(\psi\rightarrow\rho)=\nonumber\\
&&\begin{cases}
1\,\,\,\mathrm{for}\,\,\,p\leq\frac{G(\psi)}{G(\rho)},\\
\cos^{2}\left[\sin^{-1}\!\sqrt{G(\rho)}-\sin^{-1}\!\sqrt{\frac{G(\psi)}{p}}\right]\,\,\mathrm{otherwise}. \label{eq:optimal_fidelity_entanglement}
\end{cases}
\end{eqnarray}
Note that the above result also holds if the initial pure state $\ket{\psi}$ is of a dimension $\mathcal{C}^{m}\otimes\mathcal{C}^{n}$, where $m,n \geq 2$. Here, $G(\psi) = 1-\lambda^{\psi}_{1}$ and $\lambda^{\psi}_{1}$ is the largest Schmidt coefficient of $\ket{\psi}$ \cite{vidal_prob}. We refer to Appendix \ref{app_fid_prob_two_qubit} for more details.

In the setup studied so far, we allowed for an error in the final state, aiming for the optimal conversion probability with a given error. We will now go one step further, and allow an error also in the initial state. In particular, we are interested in the maximal probability of conversion $\rho \rightarrow \sigma$ allowing for an error in the initial state, i.e., the true initial state $\rho'$ fulfills $F(\rho,\rho') \geq f_1$, and similarly $F(\sigma,\sigma') \geq f_2$ for the true final state $\sigma'$. The maximal conversion probability $\mathcal{P}_{f_{1},f_{2}}$ can be defined as $\mathcal{P}_{f_{1},f_{2}}(\rho\rightarrow\sigma)=\max\left\{ P_{f_{2}}(\rho'\rightarrow\sigma):F(\rho,\rho')\geq f_{1}\right\}$, where $P_{f_2}$ is the optimal probability defined in Eq.~(\ref{eq:ProbabilityPSC}). In the
following, we give the optimal probability $\mathcal{P}_{f_1,f_2}$ for the two-qubit setting if the initial state is pure:
\begin{eqnarray}\label{optprob2}
&&\mathcal{P}_{f_1,f_2}(\psi\rightarrow\rho) =\nonumber\\
&&\begin{cases}
1\,\,\,\mathrm{for}\,\,\,m_2 \geq 0,\\
\frac{\sin^{2}\left(\sin^{-1}\sqrt{G(\psi)} + \cos^{-1}\sqrt{f_1}\right)}{\sin^{2}\left(\sin^{-1}\sqrt{G(\rho)} - \cos^{-1}\sqrt{f_2}\right)}\,\,\,\mathrm{otherwise},
\end{cases}
\end{eqnarray}
where, $m_2=\sin^{-1}\sqrt{G(\psi)}-\sin^{-1}\sqrt{G(\rho)} + \cos^{-1}\sqrt{f_1} +\cos^{-1}\sqrt{f_2}$.
We refer to Appendix \ref{app_initial_state_fidelity_ball} for the
proof. This gives a complete solution for stochastic and approximate state conversion of the two-qubit systems when the initial state is pure. 

We will now demonstrate that the methods presented above can provide new insights about our understanding of two-qubit entanglement. In particular, the developed tools allow us to estimate the distance between a two-qubit state $\rho$ and the set of states with bounded entanglement. In particular, let us consider states $\sigma$ with the property that $G(\sigma) \leq g$ with some $g \geq 0$. We are interested in the minimal distance between a given state $\rho$ and all such states $\sigma$, where the distance $D(\rho,\sigma)$ is quantified via the Bures angle: 
\begin{equation}
    D(\rho, \sigma) = \cos^{-1}\left(\sqrt{F(\rho, \sigma)}\right).
\end{equation}
In the following, we show that the minimal distance is given by
\begin{eqnarray}\label{distance_geometric}
&&\min_{G(\sigma)\leq g}D(\rho,\sigma)=\nonumber\\
&&\begin{cases}
0 & \mathrm{if}\,\,\,G(\rho)\leq g,\\
\sin^{-1}\sqrt{G(\rho)}-\sin^{-1}\sqrt{g} & \mathrm{otherwise}.
\end{cases}
\end{eqnarray}
Furthermore, we also show that these results can be extended to other entanglement quantifiers, such as entanglement of formation and concurrence.

Here, we define a new quantity named as \emph{generalized geometric entanglement} 
\begin{equation}\label{generalised_geometric}
    G_g(\rho)=1-\max_{\sigma \in S_g}F(\rho, \sigma),
\end{equation}
where $S_g$ is the set of states with geometric entanglement at most $g$. Note that $G_g$ is nonnegative, and vanishes if and only if $G(\rho) \leq g$. Moreover, $G_g$ is an entanglement monotone, i.e., 
\begin{equation}
    G_g(\Lambda[\rho])\leq G_g(\rho) \label{eq:GkMonotonicity}
\end{equation}
for any LOCC operation $\Lambda$. For this, let $\sigma \in S_g$ be such that $G_g(\rho) = 1 - F(\rho,\sigma)$. Using the fact that fidelity does not decrease under quantum operations it follows: $F(\Lambda[\rho],\Lambda[\sigma]) \geq F(\rho,\sigma)$. Since geometric entanglement does not increase under LOCC, it further holds that $\Lambda[\sigma] \in S_g$, and thus $G_g(\Lambda[\rho]) \leq 1 - F(\Lambda[\rho],\Lambda[\sigma])$. Combining these results we arrive at Eq.~(\ref{eq:GkMonotonicity}) as claimed.

In the following, we will give a closed expression for the generalized geometric entanglement for all two-qubit states, which will also provide the proof for the result given in Eq.~(\ref{distance_geometric}).

For a probability of transformation $p\geq P(\psi\rightarrow\rho)$, the optimal achievable fidelity (with $\rho$) is given by maximising the fidelity of $\rho$ with the set of states having $G\leq \frac{G(\psi)}{p}$ (because all the states in this set, can be achieved from $\psi$ with a probability $p$). Therefore, from Eq. (21) of the main text, by using $\frac{G(\psi)}{p}=g$, we will get an analytical expression for $F_g(\rho)=1-G_g(\rho)$. Hence,
\begin{eqnarray}\label{eq:min_distance}
&&G_g(\rho)=\nonumber\\
&&\begin{cases}
0 & \mathrm{if}\,\,G(\rho)\leq g,\\
\sin^2\left(\sin^{-1}\sqrt{G(\rho)}-\sin^{-1}\sqrt{g}\right) & \mathrm{otherwise}.
\end{cases}\nonumber\\
\end{eqnarray}
Using the equation $G_g(\rho)=\min_{\sigma\in S_g}\sin^2 D(\rho,\sigma)$, one attains the result given in Eq. (\ref{distance_geometric}). Here, $S_g$ denotes the set of states with $G\leq g$. 

Recalling that the geometric entanglement $G$ has a closed expression for all two qubit states, the above result allows to evaluate $G_g$ for all two-qubit states. This result also gives significant insights into the geometry of two-qubit entanglement. In particular, it allows to evaluate the distance from a given state $\rho$ to any set of states $\sigma$ with constant entanglement of formation, if the distance is measured with the Bures distance or the Bures angle.

Let us now take an arbitrary entanglement measure $M$, such that $M(\rho)=f(G(\rho))$, where $f$ is a monotonically increasing function of $G(\rho)$. Then from Eq. (\ref{eq:min_distance}), we can write   
\begin{equation}
M_k(\rho)=
\begin{cases}
0 \,\, \mathrm{if}\,\,M(\rho)\leq k,\\
\begin{multlined}  \sin^2\Big(\sin^{-1}\sqrt{f^{-1}(M(\rho))}- \\ \sin^{-1}\sqrt{f^{-1}(k)}\Big)  \,\,\mathrm{otherwise}.
\end{multlined}
\end{cases}
\end{equation}

As examples, entangled of formation $(E_F)$ and concurrence $(C)$ have a close relation with geometric measure of entanglement as follows \cite{Wei_2003,wooters_entform,Alex_Linking}
\begin{eqnarray}
&&E_F(\rho)=h(G(\rho))\\
&&C(\rho)=\sqrt{1-[1-2G(\rho)]^2},
\end{eqnarray}
where $h(x)=-x \log_2 x-(1-x)\log_2(1-x)$ and $G(\rho)\in[0,1/2]$.

\medskip

\section{Conclusions} We investigated the problem of converting quantum states within general quantum resource theories, and within the theory of entanglement. In particular, we considered probabilistic transformations, allowing for a small error in the final state. For general resource theories, we obtained upper bounds on the conversion probability and fidelity in Theorem~\ref{thm:GeneralBound}. These results significantly improve previously known bounds, and establish limits on the possible precision of probabilistic transformations in all quantum resource theories. As an application, we show that these bounds imply an upper bound on the asymptotic rates for various classes of states. This upper bound on rates turns out to be robust against sub-exponentially decaying (in number of copies of the state) probability of success and sub-exponential (in number of copies of the state) resource generation power of the transformation. We also show that the deterministic version of the single copy bounds can be applied for resource theories of quantum channels, which provide upper bounds on the conversion fidelity. In Theorem~\ref{thm:pureQuditConversion} we provide a complete solution for the stochastic-approximate transformations of arbitrary dimensional bipartite pure entangled states. In Theorem~\ref{thm:QubitConversion} we focused on two-qubit systems, and provided a complete solution to this problem if the initial state is pure. Furthermore, we apply the channel bound (Eq. (\ref{channelbound_main2})) to the resource theory of nonlocality, providing an upper bound for the optimal achievable fidelity between PR box and isotropic box. We show that this bound is tight and any locality preserving superchannel cannot increase the fidelity between the isotropic box and PR box. 



\section*{Acknowledgements} We thank Marek Miller, Manfredi Scalici, Roberto Rubboli and Ray Ganardi for discussion. We also thank Ryuji Takagi, Bartosz Regula, Benjamin Desef and Chong-long Liu for useful comments on our manuscript. This work was supported by the National Science Centre Poland (Grant No. 2022/46/E/ST2/00115). Tulja Varun Kondra also acknowledges IDUB micro-grant under Action IV.4.1 "A complex programme of support for UW PhD students" implemented in the program "Excellence Initiative - Research University".

\appendix

\section{Proof of Theorem \ref{thm:GeneralBound_resource_generating}}\label{app_theorem1}

Due to the definition of the resource robustness, there exists a state $\tau$, such that
\begin{equation}\label{robustness_decomposition}
    \frac{\rho+R(\rho)\tau}{1+R(\rho)} = \rho_f, 
\end{equation}
where $\rho_f$ is a free state. Note that, $R(\rho)<\infty$ for all $\rho$, if there exists a free state with full rank. Applying $\mathcal{E}$ on both sides of the Eq.~(\ref{robustness_decomposition}), we obtain
\begin{equation}
  \frac{1}{1+R(\rho)}\mathcal{E}( \rho)+\frac{R(\rho)}{1+R(\rho)}\mathcal{E}(\tau) = \mathcal{E}(\rho_f). \label{eq:generalBoundProof}
\end{equation}
Since we are interested in transformations with nonzero probability, we have $\tr{\mathcal{E}}(\rho) > 0$, which also implies $\tr{\mathcal{E}}(\rho_f) > 0$. Therefore,
\begin{equation}
  \frac{1}{1+R(\rho)}\frac{\mathcal{E}( \rho)}{\tr{\mathcal{E}}(\rho_f)}+\frac{R(\rho)}{1+R(\rho)}\frac{\mathcal{E}(\tau)}{\tr{\mathcal{E}}(\rho_f)} = \frac{\mathcal{E}(\rho_f)}{\tr{\mathcal{E}}(\rho_f)}.
\end{equation}
Using these results, we can evaluate the fidelity between $\FPT(\rho_f)/\tr[\FPT(\rho_f)]$ and $\sigma$:
\begin{align}
& F\left(\frac{\mathcal{E}(\rho_{f})}{\tr{\mathcal{E}}(\rho_{f})},\sigma\right)  =\nonumber\\& F\left(q\frac{\mathcal{E}(\rho)}{\tr{\mathcal{E}}(\rho)}+(1-q)\frac{\mathcal{E}(\tau)}{\tr{\mathcal{E}}(\tau)},\sigma\right), \label{eq:generalBoundProof2}
\end{align}
where we defined 
\begin{equation}
q=\frac{\tr{\mathcal{E}}(\rho)}{\tr{\mathcal{E}}(\rho_{f})[1+R(\rho)]}. \label{eq:generalBoundProof3}
\end{equation}
From Eq.~(\ref{eq:generalBoundProof}) we further see that 
\begin{equation}
1-q=\frac{\tr{\mathcal{E}}(\tau)R(\rho)}{\tr{\mathcal{E}}(\rho_{f})[1+R(\rho)]}.
\end{equation}
Using Eq.~(\ref{eq:generalBoundProof2}) and concavity of fidelity \cite{wilde_2017} we obtain 
\begin{equation}
F\left(\frac{\mathcal{E}(\rho_{f})}{\tr{\mathcal{E}}(\rho_{f})},\sigma\right)\geq qF\left(\frac{\mathcal{E}(\rho)}{\tr{\mathcal{E}}(\rho)},\sigma\right).
\end{equation}
Since,
\begin{equation}
    F^{M, \varepsilon}_{\max}(\sigma)=\max_{\tau} F(\sigma,\tau)\,\,\text{s.t}\,\,M(\tau)\leq \varepsilon
\end{equation}
and $\mathcal{E}$ is an $\varepsilon$-resource generating operation,
\begin{equation}
    q F\left(\frac{\mathcal{E}(\rho)}{\tr{\mathcal{E}(\rho)}},\sigma\right)\leq F^{M, \varepsilon}_{\max}(\sigma)
\end{equation}
In the final step we use Eq.~(\ref{eq:generalBoundProof3}), and recall that $\tr \FPT(\rho) = p$, thus arriving at
\begin{align}\label{eq:generalbound4}
\left[1+R(\rho)\right]\times F^{M, \varepsilon}_{\max}(\sigma) & \geq\frac{\tr\FPT(\rho)}{\tr{\mathcal{E}(\rho_{f})}}F\left(\frac{\mathcal{E}(\rho)}{\tr(\mathcal{E}(\rho))},\sigma\right) \\
 & \geq pF\left(\frac{\mathcal{E}(\rho)}{\tr(\mathcal{E}(\rho))},\sigma\right), \nonumber 
\end{align}
where we also used the fact that $\tr \FPT (\rho_f) \leq 1$. This completes the proof of Theorem \ref{thm:GeneralBound_resource_generating} in the main text. Note that, throughout this proof we never used the convexity of the set of free states.

\section{Asymptotic behaviour of the bound}\label{app_asymptotic}
In this section, we give proof for Theorem \ref{thm:asymptotic} in the main text. Let us first show that $F^{M, \varepsilon}_{\max}(\sigma)\leq (1+\varepsilon)F_{\max}(\sigma)$, where $M$ is either standard ($S$) or generalised ($R$) robustness. This follows from the fact that every state $\tau$, having robustness at most $\varepsilon$, can be written in the following way 
\begin{equation}\label{epsilon}
    \frac{\tau+\varepsilon \tau'}{1+\varepsilon}=\tau_{\textrm{free}}
\end{equation}
Note that, when $\varepsilon$ is standard robustness ($S$), $\tau'$ has to be chosen to be a free state. Therefore, from concavity of fidelity, it follows that
\begin{equation}
    F(\sigma,\tau)\leq F(\sigma, \tau_{\textrm{free}})(1+\varepsilon)\leq F_{\max}(\sigma)(1+\varepsilon).
\end{equation}
Since this holds for every $\tau$ such that Eq. (\ref{epsilon}) is true,
\begin{equation}
   F^{M, \varepsilon}_{\max}(\sigma)\leq (1+\varepsilon)F_{\max}(\sigma).
\end{equation}
We now use the above inequality and our bound in Eq. (\ref{eq:Main3}), to give
\begin{eqnarray}
    && F^{M, \varepsilon}_{p}\left(\rho\rightarrow\sigma\right) \leq\nonumber\\
    &&\min\left\{\frac{1}{p}\left[1+R(\rho)\right]\times (1+\varepsilon)F_{\max}(\sigma),1\right\}.
\end{eqnarray}
Rewriting the above equation in alternative form gives us
\begin{equation}
    F^{M, \varepsilon}_{p}\left(\rho\rightarrow\sigma\right)\leq \min\left\{\frac{1+\varepsilon}{p}\left[ 2^{E_{\max}(\rho) -E_{1/2}(\sigma)} \right],1\right\}
\end{equation}
Let's now assume that $E_{1/2}(\sigma^{\otimes n})=nE_{1/2}(\sigma)$ (additive). Let us also note that $E_{\max}$ is sub-additive i.e.,
\begin{equation}
    E_{\max}(\rho^{\otimes n})\leq nE_{\max}(\rho).
\end{equation}
We will now consider the case when we aim to achieve $\sigma$ with a rate $r$. In such a case, the achievable fidelity of transformation is bounded by
\begin{eqnarray}
    && F^{M,\varepsilon_n}_{p_n}\left(\rho^{\otimes n}\rightarrow\sigma^{\otimes \lfloor rn\rfloor}\right) \nonumber\\
    &&\leq \min\left\{\frac{1+\varepsilon_n}{p_n}\left[ 2^{n E_{\max}(\rho) -\lfloor rn\rfloor E_{1/2}(\sigma)} \right],1\right\}\nonumber\\
    &&\leq \min\left\{\frac{1+\varepsilon_n}{p_n}\left[ 2^{n E_{\max}(\rho) -(rn-1) E_{1/2}(\sigma)} \right],1\right\}\nonumber\\
   && = \min\left\{\frac{1+\varepsilon_n}{p_n}\left[ 2^{n E_{\max}(\rho) -rn E_{1/2}(\sigma)+E_{1/2}(\sigma) } \right],1\right\}.\nonumber\\
\end{eqnarray}
The second inequality follows from the fact that $\lfloor x\rfloor\geq x-1$. We now choose, $r=\frac{E_{\max}(\rho)}{E_{1/2}(\sigma)}+\delta$, where $\delta>0$.
\begin{eqnarray}\label{exponential}
    && F^{M,\varepsilon_n}_{p_n}\left(\rho^{\otimes n}\rightarrow\sigma^{\otimes \lfloor rn\rfloor}\right) \nonumber\\
    &&\leq \min\left\{\frac{1+\varepsilon_n}{p_n}\left[ 2^{(1-\delta \cdot n) E_{1/2}(\sigma)} \right],1\right\}.\nonumber
\end{eqnarray}
It is easy to see that $2^{(1-\delta\cdot n) E_{1/2}(\sigma)}$ goes to zero exponentially in $n$, for every $\delta>0$. This implies, whenever we allow for the generation of small amounts (sub-exponential in $n$) of resource (standard or generalised robustness) and the success probability does not decay exponentially in $n$, the the upper bound on fidelity can be made arbitrarily small by choosing big enough $n$.


\section{Deterministic bound for channels}\label{app_channel_bound}
Let $\Lambda_{all}$ be the set of all quantum channels allowed in a given setting. Transformations acting on quantum channels are described by superchannels \cite{superchannels}, with the following property $S_{all}:=\{\theta:\Lambda_{all}\rightarrow\Lambda'_{all}\}$. Here, $\Lambda'_{all}$ are the set of allowed channels in the output space. Now we define the set of allowed channels (and superchannels) as ``free channel'' (and free superchannels). It is a necessity that free superchannels map the set of free channels to themselves. Since we are interested in fundamental bounds on the transformation of channels, we define the set of free superchannels with the above-mentioned property.  Fidelity between two quantum channels $\mathcal{E}_1$ and $\mathcal{E}_2$ is defined as
\begin{equation}\label{channel_fidelity}
    F(\mathcal{E}_1,\mathcal{E}_2)=\min_{\rho}F(\openone\otimes\mathcal{E}_1(\rho),\openone\otimes\mathcal{E}_{2}(\rho)).
\end{equation}
Note that, the fidelity between quantum channels can be computed via SDP \cite{Yuan_2017}. Now, we can go ahead and define the geometric measure of a channel as 
\begin{equation}
    G(\mathcal{E})= 1-\max_{\mathcal{E}_f} F(\mathcal{E},\mathcal{E}_f),
\end{equation}
where $\mathcal{E}_f$ corresponds to a free channel. The robustness of a quantum channel is defined as \cite{Discrimination_grt}
\begin{equation}
   R(\mathcal{E})=r\in \mathbb{R}_+| \inf_{r,\mathcal{M}} : \frac{\mathcal{E}+r\mathcal{M}}{1+r} \in \Lambda_f,\,\, \mathcal{M}\in \Lambda_{\textrm{all}}
\end{equation}
Due to the definition of the resource robustness, there exists a channel $\mathcal{M}$ and a free channel $\mathcal{E}_f$, such that
\begin{equation}\label{robustness_decomposition_channel}
    \frac{\mathcal{E}+R(\mathcal{E})\mathcal{M}}{1+R(\mathcal{E})} = \mathcal{E}_f. 
\end{equation}
Again note that, $R(\mathcal{E})<\infty$ for all $\mathcal{E}$, if there exists a free channel whose Choi matrix has full rank. Let $\Theta$ be a $\varepsilon$-resource ($M$) generating super-channel. Which is defined as follows
\begin{eqnarray}
   M( \Theta(\mathcal{E}))\leq \varepsilon.
\end{eqnarray}
Applying $\Theta$  on both sides of the Eq.~(\ref{robustness_decomposition_channel}), we obtain
\begin{equation}
  \frac{1}{1+R(\mathcal{E})}\Theta( \mathcal{E})+\frac{R(\mathcal{E})}{1+R(\mathcal{E})}\Theta(\mathcal{M}) = \Theta(\mathcal{E}_f). \label{eq:generalBoundProofChannel}
\end{equation}
Therefore, it follows that
\begin{eqnarray}
   && \min_{\rho}F\left(\openone\otimes\Theta(\mathcal{E}_{f})(\rho),\openone\otimes\mathcal{E}_t(\rho)\right) = \nonumber\\
    &&\min_{\rho}F\Big(q\openone\otimes\Theta(\mathcal{E})(\rho)+\nonumber\\
    &&\qquad\quad\,(1-q)\openone\otimes\Theta(\mathcal{M})(\rho),\openone\otimes\mathcal{E}_t(\rho)\Big), \label{eq:generalBoundProof2Channel}
\end{eqnarray}
where  
\begin{equation}
q=\frac{1}{[1+R(\mathcal{E})]}. \label{eq:generalBoundProof3Channel}
\end{equation}
From Eq.~(\ref{eq:generalBoundProofChannel}) we further see that 
\begin{equation}
1-q=\frac{R(\mathcal{E})}{[1+R(\mathcal{E})]}.
\end{equation}
Using Eq.~(\ref{eq:generalBoundProof2Channel}) and concavity of fidelity \cite{wilde_2017} we obtain 
\begin{align}\label{eq:generalBoundProof3Channel}
&\min_{\rho}F\left(\openone\otimes\Theta(\mathcal{E}_{f})(\rho),\openone\otimes\mathcal{E}_t(\rho)\right) \geq\\ \nonumber
&q\min_{\rho}F\left(\openone\otimes\Theta(\mathcal{E})(\rho),\openone\otimes\mathcal{E}_t(\rho)\right).
\end{align}
From the definition of $\varepsilon$-resource generating super-channel, it follows that,  $\min_{\rho}F\left(\openone\otimes\Theta(\mathcal{E}_{f})(\rho),\openone\otimes\mathcal{E}_t(\rho)\right) \leq F^{M, \varepsilon}_{\max}(\mathcal{E}_t)$. Here, $F^{M, \varepsilon}_{\max}(\mathcal{E})$ is defined as
\begin{eqnarray}
    && F^{M, \varepsilon}_{\max}(\mathcal{E})=\nonumber\\
    &&\max_{\mathcal{M}_1:M(\mathcal{M}_1)\leq\varepsilon}\min_{\rho}F\left(\openone\otimes\mathcal{M}_1(\rho),\openone\otimes\mathcal{E}(\rho)\right).
\end{eqnarray}
Using this in Eq. (\ref{eq:generalBoundProof3Channel}), we find
\begin{equation}
F^{M, \varepsilon}_{\max}(\mathcal{E}_t)\geq q\min_{\rho}F\left(\openone\otimes\Theta(\mathcal{E})(\rho),\openone\otimes\mathcal{E}_t(\rho)\right).
\end{equation}
Note that, from Eq. (\ref{channel_fidelity})
\begin{equation}
    \min_{\rho}F\left(\openone\otimes\Theta(\mathcal{E})(\rho),\openone\otimes\mathcal{E}_t(\rho)\right)=F\left(\Theta(\mathcal{E}),\mathcal{E}_t\right).
\end{equation}
Therefore,
\begin{align}\label{channelbound}
\min\{\left[1+R(\mathcal{E})\right]\times F^{M, \varepsilon}_{\max}(\mathcal{E}_t),1\}\geq F\left(\Theta(\mathcal{E}),\mathcal{E}_t\right),
\end{align}
holds for every $\varepsilon$-resource generating super-channel ($\Theta$). This completes the proof. When $\Theta_f$ is a resource non-generating super-channel ($\varepsilon=0$), we get
\begin{align}
\min\{\left[1+R(\mathcal{E})\right]\times F_{\max}(\mathcal{E}_t),1\}\geq F\left(\Theta_f(\mathcal{E}),\mathcal{E}_t\right).
\end{align}
This can be alternatively written as
\begin{eqnarray}
&& F\left(\Theta_f(\mathcal{E}),\mathcal{E}_t\right)
\nonumber\\
&&\leq \min\left\{ (1+R(\mathcal{E}))\times( 1-G(\mathcal{E}_t)) ,1\right\}.
\end{eqnarray}
Here, $G(\mathcal{E}_t)$ is the geometric measure of the channel $\mathcal{E}_t$. In the case when $\mathcal{E}=\mathcal{E}_t$, we get
\begin{equation}
    1\leq \min\left\{ (1+R(\mathcal{E}))\times( 1-G(\mathcal{E})) ,1\right\}.
\end{equation}
The above equation can be alternatively written as
\begin{equation}
    1+R(\mathcal{E})\geq\frac{1}{1-G(\mathcal{E})}.
\end{equation}
This is the generalisation of the bound given in Eq. (\ref{r,g_bound}), to the case of quantum channels.
\section{Proof of Eq. (\ref{eq:distill_singlet_locc})}\label{app_distill_locc}
It has been shown \cite{Verstraete1} that the optimal achievable fidelity to a singlet when starting from an arbitrary two-qubit state $\rho$ can be achieved by using a 1-LOCC (one way LOCC) protocol and is connected to the robustness of entanglement \cite{Vidal_robustness}. The optimal LOCC protocol can be given by the following procedure. With a probability $p_o$ Alice and Bob can achieve a two-qubit state which has fidelity $F_{p_o}$ with the maximally entangled state and in case of a failure they create a pure separable state $\ket{00}\!\bra{00}$. Therefore, 
\begin{eqnarray}
    \max_{\Lambda\in \textrm{LOCC}}F(\Lambda(\rho),\phi_{+})&&= p_oF_{p_o}+\frac{1-p_{o}}{2}\nonumber\\
    &&=\frac{1+R(\rho)}{2},\label{LOCC bound}
\end{eqnarray}
where $R(\rho)$ represents the robustness of entanglement for the state $\rho$ \cite{Vidal_robustness}. This is equivalent to 
\begin{equation}
    p_o\left(F_{p_o}-\frac{1}{2}\right) +\frac{1}{2}=\frac{1+R(\rho)}{2}
    \end{equation}
and therefore
\begin{equation}
F_{p_o}=\frac{R(\rho)}{2p_o}+\frac{1}{2}.
\end{equation}
Therefore, if an SLOCC protocol has a success probability at least ``$p$'', then the optimal achievable singlet fraction satisfies the following inequality
\begin{equation}\label{bound_singlet}
    F_{p}\left(\rho\rightarrow\phi_{+}\right)\leq\min\left\{\frac{R(\rho)}{2p}+\frac{1}{2},1\right\},
\end{equation}
because,  if there exists an SLOCC protocol violating this bound then we can complete the SLOCC protocol by preparing a separable state $\ket{00}\!\bra{00}$ in case of a failure. This LOCC protocol would violate the bound given in Eq. (\ref{LOCC bound}).

\section{Proof of Theorem \ref{thm:pureQuditConversion}}\label{app_proof_pure_bipartite}

We now provide analytical expressions for $F_p$ and $P_f$ when the initial ($\psi$) and target ($\phi$) states are pure states of arbitrary finite dimension.

Let $S_p$ be the set of states which can be achieved from $\psi$ with at least a probability $p$. Therefore
\begin{equation}
    F_p(\psi\rightarrow\phi)=\max_{\rho\in S_p}\bra{\phi}\rho\ket{\phi}.
\end{equation}
We now show that the state in $S_p$ with optimum fidelity (to $\phi$) can always be chosen to be a pure state. Firstly, we define 
\begin{equation}
    E^{\psi}_l=\sum_{i=l}^{n}\alpha_i\,\,\textrm{and}\,\, E^{\phi}_l=\sum_{i=l}^{n}\beta_i,
\end{equation}
where $\{\alpha_i\}$, $\{\beta_i\}$ are the squared Schmidt coefficients of $\psi$ and $\phi$ in decreasing order. 
In \cite{Vidal_approx}, it was showed that an ensemble of quantum states $\{q_j,\rho_j\}$ can be generated from $\psi$ via LOCC if and only if there exists a pure state ensemble $\{q_kr_{k,j}, \psi_{k,j}\}$, satisfying
\begin{equation}
    E_l^{\psi}\geq \sum_{k,j}q_kr_{k,j}E^{\psi_{k,j}}_{l} \,\,\textrm{for all}\,\,l, \label{dec}
\end{equation}
where $\rho_k=\sum_j r_{k,j}\psi_{k,j}$.
This implies, to generate $\rho$ with a probability $p$, the following inequality must be satisfied for at least one pure state decomposition $\{p_k,\mu_k\}$ of the state $\rho$:
\begin{equation}
    E_l^{\psi}\geq p\sum_{k} p_kE^{\mu_k}_{l} \,\,\textrm{for all}\,\,l. \label{dec}
\end{equation}

Let the squared Schmidt coefficients of $\mu_k$ be $\{\gamma_{i}^{k}\}$, ordered in decreasing order. In the next step, we will prove the following inequality:
\begin{eqnarray}\label{upperbound}
   f = \bra{\phi}\rho\ket{\phi} &&\leq \sum_kp_k\left(\sum_{i=1}^{n}\sqrt{\gamma_{i}^{k}\beta_i}\right)^{2}\nonumber\\
   &&\leq \left(\sum_{i=1}^{n}\sqrt{\sum_kp_k\gamma_{i}^{k}\beta_i} \right)^{2}.
\end{eqnarray}
The first inequality follows from the fact that the fidelity of each pure state (in the decomposition) to $\phi$, is maximal when they have all same Schmidt basis as $\phi$ \cite{Vidal_approx}. In order to see the second inequality, let us define diagonal matrices $\tau_\phi$, $\tau_{\mu_{k}}$ and $\tau_\chi$, with diagonal elements as the squares of ordered Schmidt coefficients of $\phi$, $\mu_{k}$ and $\chi$ respectively. $\tau_\chi$ is defined as 
\begin{equation}
     \tau_\chi=\sum_{k}{p_k}\tau_{\mu_{k}}.
\end{equation}
 The second inequality then follows from concavity of fidelity: 
 \begin{equation}
     \sum_{k}{p_k}F(\tau_{\mu_{k}},\tau_\phi)\leq F(\tau_\chi,\tau_\phi).
 \end{equation}

Now we define a pure state, the Schmidt coefficients of $\chi$ are square-roots of the diagonal elements of $\tau_{\chi}$
\begin{equation}
 \ket{\chi}= \sum_{i=1}^{n}\sqrt{\sum_k p_k\gamma^{k}_{i}}\ket{i_Ai_B}\label{state}
\end{equation}
where, $\ket{i_Ai_B}$ are the same Schmidt basis as $\phi$. Here, note that
\begin{eqnarray}
    E_l^{\chi}=\sum_{i=l}^n\left(\sum_k p_k\gamma_i^k\right)
    &&=\sum_k p_k \sum_{i=l}^n\gamma_i^k\nonumber\\
    &&=
    \sum_{k} p_kE^{\mu_k}_{l}.\label{sum}
\end{eqnarray}
 It is easy to see that, 
\begin{equation}
 | \langle \chi| \phi \rangle|^{2}  =  \left(\sum_{i=1}^{n}\sqrt{\sum_kp_k\gamma_{i}^{k}\beta_i} \right)  ^{2}
\end{equation}
coinciding with the upper bound in Eq. (\ref{upperbound}), on $f$. Since, the Schmidt coefficients of $\chi$ are square-roots of the diagonal elements of $\tau_{\chi}$, From Eq. (\ref{dec}) and Eq. (\ref{sum}), it follows that
 \begin{equation}
    E_l^{\psi}\geq p\sum_{i} p_iE^{\mu_i}_{l} = p E_l^{\chi} \,\,\textrm{for all}\,\,l.
\end{equation}
Therefore, from \cite{vidal_purestates}, $\chi$ can be achieved via SLOCC from $\psi$ as the initial state, with at least a probability $p$. This shows that  the state in $S_p$ with optimum fidelity (to $\phi$) can always be chosen to be a pure state. In the following, we construct a pure state in $S_p$, which has optimal fidelity with $\phi$. 

For the probability distribution $\vec{\beta}$ (decreasing order and dimension $n$), from \cite{steepest}, we know that, there exits a $\vec{\beta}^{\delta}$ satisfying $\norm{\vec{\beta}-\vec{\beta}^{\delta}}_1\leq \delta$ (where $\norm{\vec{\gamma}}$ is defined as $\norm{\vec{\gamma}}=\sum_i|\gamma_i|$ ) and majorize every other state in the $\delta$-ball (steepest $\delta$-approximation of $\vec{\beta}$). The authors also provide a method to construct such a distribution, which goes as follows. If $(1,0, \ldots, 0)$ lies inside the $\delta$-ball, $\vec{\beta}^{\delta}=(1,0, \ldots, 0)$. Otherwise, we add $\frac{\delta}{2}$ to $\beta_1$ and choose an integer $l^{*} \in\{1, \ldots, n\}$ such that
\begin{equation}
\beta_1+\frac{\delta}{2} + \sum_{i=2}^{l^{*}} \beta_{i} \leqslant 1 \quad \text {and} \quad \beta_1+\frac{\delta}{2} +\sum_{i=2}^{l^{*}+1} \beta_{i}>1.
\end{equation}
The steepest $\delta$-approximation of $\vec{\beta}$ is
\begin{equation}\label{construction}
\vec{\beta}^{\delta}= \begin{cases}\beta_1+\frac{\delta}{2} & \text { for } \quad i=1\\
\beta_{i} & \text { for } \quad 1<i<l^{*}+1 \\ 1-x & \text { for } \quad i=l^{*}+1 \\ 0 & \text { for } \quad i>l^{*}+1\end{cases}
\end{equation}
with $x=\beta_1+\frac{\delta}{2}+\sum_{i=2}^{l^{*}}\beta_i$. In this way we can construct a pure state $\phi_{st}$ from $\phi$, given by
\begin{equation}
   \ket{\phi_{st}} = \sum_i^{n}\sqrt{\beta_{i}^{\delta}}\ket{i_Ai_B},
\end{equation}
where, $\delta=\sqrt{1-f_1}$. By construction, the fidelity of $\phi_{st}$ with $\phi$ is $f_1$ and $\phi_{st}$ can be achieved deterministically from every pure state which has a fidelity at least $f_1$ (with $\phi$). Now we need to find the optimal $f_1$ such that $\phi_{st}$ can be achieved from $\psi$ with at least $p$ probability. This implies that,  $F_p(\psi\rightarrow\phi)$ is given by the maximum value of $f_1$, such that $\phi_{st}$ can be achieved with at least probability $p$, from $\psi$. Equivalently, $F_p(\psi\rightarrow\phi)$ is the maximum possible value of $f_1$ satisfying \cite{vidal_purestates},
\begin{eqnarray}\nonumber
p\times\max&&\left\{\sum_{i=l}^{n}\beta_i- \frac{\sqrt{1-f_1}}{2},0\right\}\nonumber\\
&&\leq \sum_{i=l}^{n}\alpha_i\,\,\textrm{for}\,\, l:2,\cdots, n.
\end{eqnarray}
Since, $\alpha_i$'s are all non-negative, we can drop the ``$\max$'' from the above inequality and we have 
\begin{eqnarray}
&&p(E_{l}^{\phi}- \frac{\sqrt{1-f_1}}{2})\leq E_{l}^{\psi}\, \textrm{for}\, l:2,\cdots,n.
\end{eqnarray}
Solving this gives
\begin{equation}
    f_1\leq1-4\left( E_{l}^{\phi}-\frac{ E_{l}^{\psi}}{p}\right)^2\,\,\textrm{for}\,\, l:2,\cdots,n.
\end{equation}
Therefore,
\begin{eqnarray}
    &&F_p(\psi\rightarrow\phi)\\
    &&=\min_{l}\left\{1-4\left( E_{l}^{\phi}-\frac{ E_{l}^{\psi}}{p}\right)^2\right\}\,\,\textrm{for}\,\, l:2,\cdots,n.\nonumber
\end{eqnarray}

Now using the results in \cite{vidal_purestates}, we provide a closed expression for $P_f(\psi\rightarrow\phi)$. In this case, the optimal transformation can be achieved by transforming $\ket{\psi}$ into a pure state given by $\ket{\phi^{\sqrt{1-f}}_{st}}$. Therefore
\begin{eqnarray}
    &&P_f(\psi\rightarrow\phi)\\
    &&=\min\left\{1,\min_l \left\{\frac{E_l^{\psi}}{\max\left\{E_l^{\phi}- \frac{\sqrt{1-f}}{2},0\right\}}\right\} \right\},\nonumber
\end{eqnarray}
where $l$ ranges from 2 to $n$.


\section{Proof for Lemma \ref{thm:TwoQubitsContinuity}}\label{app_tight_continuinty}

In this section, we give tight continuity bounds on the geometric entanglement in the two-qubit settings. These bounds give new insight into the properties of two-qubit entanglement and will be used to prove Theorem~\ref{thm:QubitConversion} of the main text. We will make use of the following distance measure:
\begin{equation}
    D(\rho, \sigma) = \cos^{-1}\left(\sqrt{F(\rho, \sigma)}\right), \label{eq:BuresAngle}
\end{equation}
which is also known as the Bures angle. $D(\rho,\sigma)$ is a proper distance measure for the set of quantum states, as it is non-negative and zero if and only if $\rho = \sigma$, and it satisfies the triangle inequality. Note that for any state of two qubits the geometric entanglement $G$ is between $0$ and $1/2$. For any states $\rho$ and $\rho' \in S_{\rho,f}$ it is thus possible to define $\{\alpha, \beta\}\in [0,\pi/4]$ and $k \in [0,\pi/2]$ such that
\begin{subequations} 
\begin{align}
G(\rho) & =\sin^{2}\alpha,\label{eq:assumptions1} \\
G(\rho') & =\sin^{2}\beta,\label{eq:assumptions2} \\
f & = \cos^{2}k.\label{eq:assumptions3}
\end{align}
\end{subequations}
Firstly, we provide bounds on the minimum and maximum achievable geometric measure in the set $S_{\rho,f}$. Let us assume $\rho_s$ is the closest separable state to $\rho$ with respect to the Bures angle. Similarly, $\rho'_s$ is the closest separable state to $\rho'$. From Eq.~(\ref{eq:assumptions1}) it follows that
\begin{eqnarray}
D(\rho,\rho_{s})&&=\cos^{-1}\left(\sqrt{F(\rho,\rho_{s})}\right)\nonumber\\
&&=\cos^{-1}\left(\sqrt{1-G(\rho)}\right)=\alpha.
\end{eqnarray}
Similarly, we obtain
\begin{align}
D(\rho',\rho'_{s}) & =\beta, \\
D(\rho,\rho') & \leq k,
\end{align}
where in the last inequality we used the fact that $\rho' \in S_{\rho,f}$. Now, using the triangle inequality and the fact that $\rho_s$ is the closest separable state to $\rho$ we obtain
\begin{eqnarray}
\alpha = D(\rho,\rho_s) \leq D(\rho,\rho'_s) &&\leq D(\rho,\rho')+D(\rho', \rho'_s) \nonumber\\
&&\leq k + \beta.
\end{eqnarray}
Therefore, we find
\begin{align}
    \beta \geq \max\{\alpha - k, 0\}. \label{lower_bound}
\end{align}
Eqs.~(\ref{eq:assumptions2}) and  (\ref{lower_bound}) imply a lower bound on $G(\rho')$:
\begin{equation}
G(\rho') \geq \sin^2(\max\{\alpha-k,0\}). \label{geometric_lower_bound}
\end{equation}
Similarly, we can also find an upper bound on $G(\rho')$. In particular, it holds:
\begin{align}
\beta = D(\rho',\rho'_s) \leq D(\rho',\rho_s) &\leq D(\rho,\rho')+D(\rho, \rho_s)\\
&\leq k + \alpha. \nonumber
\end{align}
Therefore, we have
\begin{align}
    \beta \leq \min \{ \alpha + k, \pi/4\}. \label{upper_bound}
\end{align}
From Eqs.~(\ref{eq:assumptions2}) and (\ref{upper_bound}), the upper bound on $G(\rho')$ is given by
\begin{equation}
G(\rho') \leq \sin^2(\min \{ \alpha + k, \pi/4\}). \label{geometric_upper_bound}
\end{equation}

In the following, we show that the lower bound in Eq. (\ref{geometric_lower_bound}) can be achieved. Note that for every two-qubit state $\rho$ we can find a pure state decomposition such that $\rho = \sum_i p_i \ket{\psi_i}\!\bra{\psi_i}$ and $G(\psi_i) = G(\rho)$ holds true for all states $\ket{\psi_i}$ \cite{wooters_entform,vidal_prob,Wei_2003}. This implies that each of the states $\ket{\psi_i}$ can be written as
\begin{eqnarray}
    \ket{\psi_i} &= \cos{\alpha}\ket{a_i}\ket{b_i} + \sin{\alpha}\ket{a^{\perp}_i}\ket{b^{\perp}_i},
\end{eqnarray}
where $\braket{a_i}{a^{\perp}_i} = \braket{b_i}{b^{\perp}_i} = 0$. Now we choose a state $\rho_{\min}$ such that
\begin{align}\label{minimum_decomposition}
    \rho_{\min} = \sum_i q_i \ket{\phi_i}\!\bra{\phi_i}
\end{align}
with the pure states
\begin{eqnarray}
    \ket{\phi_i} &= \cos{\tilde{\beta}}\ket{a_i}\ket{b_i}  + \sin{\tilde{\beta}}\ket{a^{\perp}_i}\ket{b^{\perp}_i}
\end{eqnarray}
with $\tilde \beta \in [0,\pi/4]$ and probabilities
\begin{equation}
    q_i = \frac{p_i |\braket{\psi_i|\phi_i}|^2}{\sum_k p_k |\braket{\psi_k|\phi_k}|^2}. \label{prob}
\end{equation}
Note that, from the convexity of geometric measure, it follows
\begin{equation}
    G(\rho_{\min})\leq \sin^2\tilde{\beta}. \label{convexity}
\end{equation}
Using the properties of fidelity \cite{wilde_2017}, we obtain 
\begin{align}
   \sqrt{F(\rho, \rho')} &\geq \sum_i \sqrt{p_i q_i} |\braket{\psi_i|\phi_i}| \nonumber \\
    &= \sqrt{\sum_i p_i |\braket{\psi_i|\phi_i}|^2} = |\cos (\alpha -\tilde{\beta})|.
    \end{align}
    Therefore,
    \begin{equation}
        F(\rho, \rho_{\min}) \geq \cos^{2}(\alpha -\tilde{\beta}). \label{fidelity}
    \end{equation}

In the following, we set
\begin{equation}
\tilde{\beta}=\max\{\alpha-k,0\}. \label{eq:btilde}
\end{equation}
For this choice Eq. (\ref{fidelity}) becomes
     \begin{equation}\label{ineq1}
        F(\rho, \rho_{\min}) \geq \cos^{2}\left(\min\{k,\alpha\}\right) \geq \cos^{2}k. \end{equation}
The above inequality holds as $\cos^2$ is a monotonically decreasing function in $[0, \pi/4]$ and $\min\{k,\alpha\}$ is in $[0, \pi/4]$. Therefore, the state $\rho_{\min}$ is within $S_{\rho,f}$. Moreover, from Eq.~(\ref{convexity}) we obtain
\begin{equation}
G(\rho_{\min})\leq\sin^{2}\left(\max\{\alpha-k,0\}\right). \label{eq:btildeUpperBound}
\end{equation}

In the next step, recall that for every state in $S_{\rho,f}$ the geometric entanglement is bounded from below by Eq. (\ref{geometric_lower_bound}). This also applies to the state $\rho_{\min}$. Since the geometric entanglement of $\rho_{\min}$ is bounded from above by Eq. (\ref{eq:btildeUpperBound}), we obtain 
\begin{equation}
 G(\rho_{\min}) =  \sin^2(\max\{\alpha-k,0\}).
 \end{equation}
This shows that the lower bound in Eq.~(\ref{geometric_lower_bound}) is achievable. Hence, the minimum geometric measure within the set $S_{\rho,f}$ is given by 
\begin{eqnarray}\label{min_geo}
&&\min_{\rho'\in S_{\rho,f}}G(\rho') = \sin^2\left(\max\{\alpha-k,0\}\right)\\
&&=\sin^{2}\left(\max\left\{ \sin^{-1}\!\sqrt{G(\rho)}-\cos^{-1}\!\sqrt{f},0\right\} \right). \nonumber
\end{eqnarray}

We will now discuss explicitly the case where the state $\rho$ is pure, i.e.,
\begin{eqnarray}
\ket{\psi}=\cos{\alpha}\ket{a}\ket{b} + \sin{\alpha}\ket{a^{\perp}}\ket{b^{\perp}}.
\end{eqnarray}
In this case we will show that the upper bound given in Eq.~(\ref{geometric_upper_bound}) is achievable as well. For this, we choose
\begin{align}
    \ket{\psi_{\max}} &= \cos(\min \{ \alpha + k, \pi/4\})\ket{a}\ket{b} \label{pure_max}\\
    &+ \sin(\min \{ \alpha + k, \pi/4\})\ket{a^{\perp}}\ket{b^{\perp}}. \nonumber
\end{align}
Note that 
\begin{align} 
G(\psi_{\max}) &= \sin^2(\min \{ \alpha + k, \pi/4\}). 
\end{align}
Furthermore, we note that
\begin{equation}
F(\psi, \psi_{\max}) = \cos^{2}\left(\min \left\{k,\frac{\pi}{4}-\alpha\right\} \right) \geq \cos^{2}k. 
\end{equation}
The above inequality holds as $\cos^2$ is a monotonically decreasing function in $[0, \pi/4]$ and $\min \{k,\pi/4-\alpha\}$ are in $[0, \pi/4]$. Hence, $\ket{\psi_{\max}}$ is inside $S_{\psi,f}$ and has the maximum possible geometric entanglement as following
\begin{eqnarray}
&&\max_{\rho'\in S_{\psi,f}}G(\rho') = \sin^2\left(\min\{\alpha+k,\pi/4\}\right)\\
&&=\sin^{2}\left(\min\left\{ \sin^{-1}\!\sqrt{G(\psi)}+\cos^{-1}\!\sqrt{f},\pi/4\right\} \right)\nonumber.
\end{eqnarray}
This completes the proof.

\section{Proof of Theorem \ref{thm:QubitConversion} }\label{app_fid_prob_two_qubit}

In \cite{vidal_prob}, it has been shown that for an initial two-qubit pure state $\ket{\psi}$, the optimal probability to achieve a two-qubit state $\rho$ with unit fidelity via LOCC is given by
\begin{equation}
    P(\psi \rightarrow \rho) =\min \{\frac{G(\psi)}{G(\rho)},1\}. \label{opt_prob}
\end{equation}

Now, if we want to go to a final state $\rho$ with some fidelity $f$, it is clear that to achieve maximal probability, we need to stochastically transform $\ket{\psi}$ into a state $\rho'$ with the minimum geometric measure of entanglement such that $F(\rho,\rho') \geq f$, and thus
\begin{equation}
P_{f}(\psi\rightarrow\rho)=\min\{\frac{G(\psi)}{\min_{\rho'\in S_{\rho,f}}G(\rho')},1\}.
\end{equation}

Recall the tight continuity bounds for geometric measure, it holds:
\begin{eqnarray}
&&\min_{\rho'\in S_{\rho,f}}G(\rho')  \\
&&=\sin^{2}\left(\max\left\{ \sin^{-1}\!\sqrt{G(\rho)}-\cos^{-1}\!\sqrt{f},0\right\} \right).\nonumber
\end{eqnarray}
Using the quantity
\begin{equation}
    m_1=\sin^{-1}\sqrt{G(\psi)}-\sin^{-1}\sqrt{G(\rho)} +\cos^{-1}\sqrt{f},
\end{equation}
let us now consider the case when $m_1\geq0$, which translates to 
\begin{equation}
    \sin^{-1}\sqrt{G(\rho)}-\cos^{-1}\sqrt{f} \leq \sin^{-1}\sqrt{G(\psi)}.
\end{equation}
It holds that 
\begin{equation}
\sin^{-1}\sqrt{G(\rho)}-\cos^{-1}\sqrt{f}\in[-\pi/2,\pi/4]
\end{equation}
and $\sin^{-1}\sqrt{G(\psi)}\in[0,\pi/4]$. Therefore, we have 
\begin{eqnarray}
    &&\max\left\{\sin^{-1}\sqrt{G(\rho)}-\cos^{-1}\sqrt{f},0\right\}\nonumber\\
    &&\leq \sin^{-1}\sqrt{G(\psi)}.
\end{eqnarray}
From these results, we obtain
\begin{eqnarray}
&&\min_{\rho'\in S_{\rho,f}}G(\rho')\nonumber\\ && =\sin^{2}\left(\max\left\{ \sin^{-1}\sqrt{G(\rho)}-\cos^{-1}\sqrt{f},0\right\} \right)\nonumber\\
 && \leq\sin^{2}(\sin^{-1}\sqrt{G(\psi)})=G(\psi). 
\end{eqnarray}
For $G(\psi) > 0$ if follows that 
\begin{equation}
    \frac{G(\psi)}{\min_{\rho'\in S_{\rho,f}}G(\rho')} \geq 1.
\end{equation}
This implies that $P_{f}(\psi\rightarrow\rho)=1$ if $m_1\geq 0 $. 

Next, we consider the case when $m_1<0$, which can be expressed as,
\begin{align}
   \sin^{-1}\sqrt{G(\rho)}-\cos^{-1}\sqrt{f}>\sin^{-1}\sqrt{G(\psi)}>0.
\end{align}
Therefore, in this case we have
\begin{equation}\label{eq:optimal_probability_entanglement}
    P_{f}(\psi\rightarrow\rho)=\frac{G(\psi)}{\sin^{2}(\sin^{-1}\sqrt{G(\rho)}-\cos^{-1}\sqrt{f}) }.
\end{equation}
This completes the proof of Theorem 3.

From the above result we can also obtain a closed expression for $F_p$. For the case, when $p \leq \frac{G(\psi)}{G(\rho)}<1$, (here, we assume $G(\psi)<G(\rho)$) from Eq. (\ref{opt_prob}), one can see that $F_p (\psi \rightarrow \rho) = 1$. Note that, if $G(\psi)\geq G(\rho)$, then the transformation is always possible exactly with unit probability \cite{vidal_prob}. When $1\geq p>\frac{G(\psi)}{G(\rho)}$, the optimal achievable fidelity can be obtained by solving Eq.~(\ref{eq:optimal_probability_entanglement}) for $f$, which gives 
\begin{eqnarray}
&&F_p (\psi \rightarrow \rho)\nonumber\\
&&=\cos^{2}\left[\sin^{-1}\!\sqrt{G(\rho)}-\sin^{-1}\!\sqrt{\frac{G(\psi)}{p}}\right].
\end{eqnarray}
Combining these results we arrive at Eq.~(16) of the main text.
\subsection{Examples to compare the exact result and the upper bound}
In the following we provide concrete examples, also comparing the optimal conversion fidelity $F_p$ with the upper bound on fidelity for general resource theories, see Theorem 1 of the main text. In all our examples the initial state will be pure: $\ket{\psi}=\cos\alpha \ket{00}+\sin \alpha \ket{11}$. In the first example (see Fig. \ref{fig:fidelity_werner}) the final target state is a two-qubit Werner state~\cite{Werner_1989}
\begin{equation}
\rho_{W}=r\ket{\phi^{+}}\!\bra{\phi^{+}}+(1-r)\frac{\openone}{4}
\end{equation}
with $\ket{\phi^+} = (\ket{00}+\ket{11})/\sqrt 2$. The geometric measure of entanglement of the Werner state is $G(\rho_W)=\max\{ \frac{1}{2} - \frac{1}{4} \sqrt{3 + 6 r - 9 r^2},0\}$ \cite{vidal_prob,Wei_2003,wooters_entform}.
Note that the state is separable for $r\leq 1/3$. It is possible to achieve exact conversion $\psi \rightarrow \rho_W$ with optimal probability $\min\{G(\psi)/G(\rho_W),1\}$ \cite{vidal_prob}. A larger probability can be achieved with stochastic approximate conversion, and the optimal fidelity can be found using Eq. (21) of the main text. In Fig. \ref{fig:fidelity_werner} (a) the optimal fidelity $F_p(\psi \rightarrow \rho_W)$ is shown as function of the probability $p$ for $\alpha=0.01$ and $r=0.9$. In the small-probability regime, it is possible to achieve a significant increase of the transformation fidelity by only a small reduction of the conversion probability. A similar behavior is found for fidelity close to $0.7718$. In this regime, a small reduction of the transformation fidelity can lead to a significant increase of the transformation probability. In the inset of Fig.~\ref{fig:fidelity_werner} (a) we also show the upper bound on $F_p$ given in Theorem 1. The gap between the upper bound and the true value of $F_p$ is nonzero for all $p$, and for $p=1$ the gap reaches its minimal value $0.0068$.

In Fig. \ref{fig:fidelity_werner} (b) we show the optimal fidelity $F_p(\psi \rightarrow \rho_W)$ as a function of the Werner state parameter $r$ for a fixed transformation probability $p=0.75$, where again $\alpha=0.01$. If $r$ is close to $1$ it is possible to achieve a significant increase of the transformation fidelity by only a small reduction of $r$. These results suggest that in some settings it might be more advantageous to establish a noisy Werner state with high fidelity, rather than aiming for the Bell state $\ket{\phi^+}$ with a large error. 

For completeness, we also show results for the conversion $\ket{\psi} \rightarrow \ket{\phi^+}$ in Fig. \ref{fig:fidelity_pure}. The main figure shows that optimal conversion fidelity $F_p(\ket{\psi} \rightarrow \ket{\phi^+})$ as a function of $p$ for $\alpha=0.2$. The inset of Fig. \ref{fig:fidelity_pure} also shows the upper bound given in Theorem 1. For $p=1$ the upper bound coincides with $F_p$. Note that for pure target states the upper bound has been investigated previously in~\cite{Regula_distillation}.

\begin{figure}
\includegraphics[width=\columnwidth]{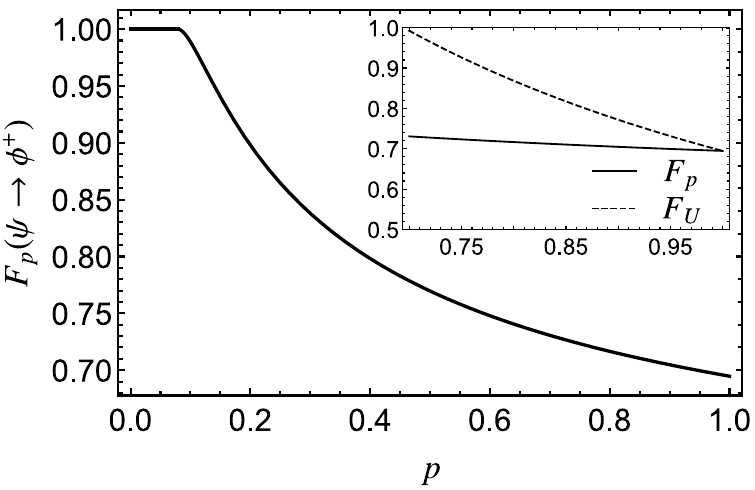}
\caption{Optimal fidelity for the transformation $\ket{\psi} \rightarrow \ket{\phi^+}$ is plotted with probability $p$,  where we choose $\alpha=0.2$. In the inset, optimal fidelity (solid line) and the upper bound of fidelity (dashed line) are shown with respect to $p$. Note that fidelity is $1$, when $p\leq G(\psi)/G(\phi^+)\approx 0.079$.}
\label{fig:fidelity_pure}
\end{figure}

\begin{figure}[h!]
\includegraphics[width=\columnwidth]{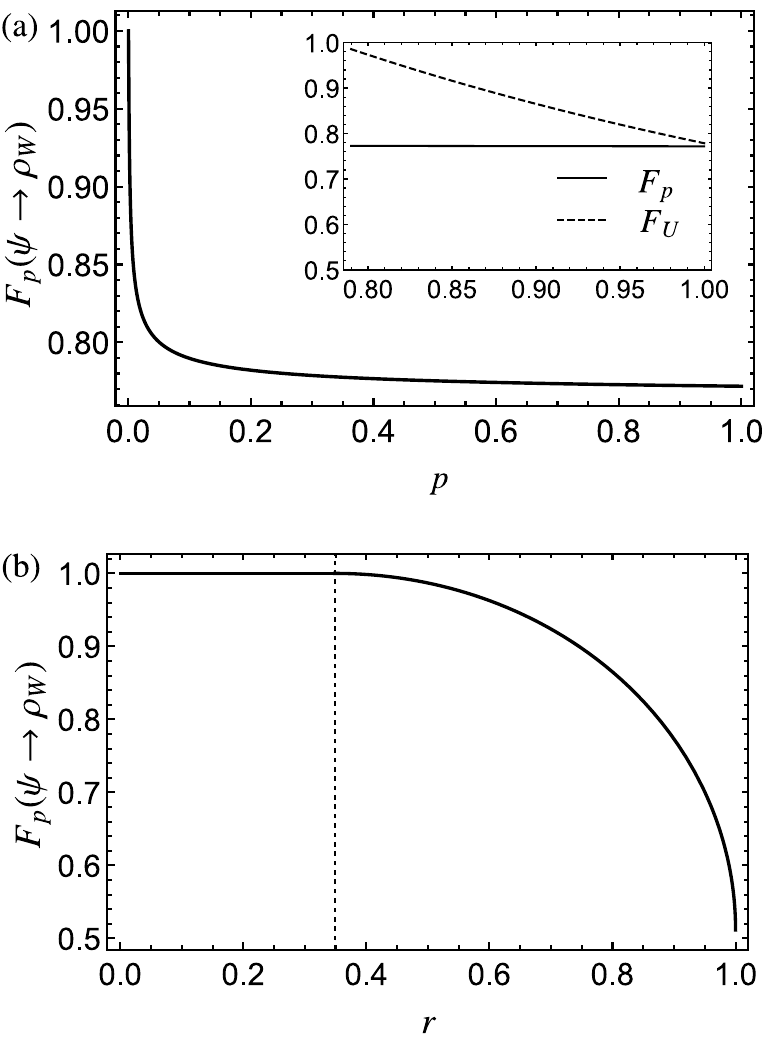}
\caption{(a) Optimal fidelity for the transformation $\psi \rightarrow \rho_W$ is plotted with probability $p$, where we choose $\alpha=0.01$ and $r=0.9$. In the inset, optimal fidelity (solid line) and the upper bound of fidelity (dashed line) are shown with respect to $p$. For $p=1$, there is a very small gap of $0.0068$ between the upper bound of fidelity and the exact fidelity. Optimal fidelity is $1$ for $p\leq G(\psi)/G(\rho_W)\approx 0.0004$, which is not visible in the figure.
(b) Optimal fidelity is plotted with $r$, where we consider $\alpha=0.01$ and transformation probability to be $0.75$. The vertical dashed line is at $r\approx 0.3487$ and below or equal this $r$ we have unit fidelity.}
\label{fig:fidelity_werner}
\end{figure}

\section{Proof of Eq. (\ref{optprob2})}\label{app_initial_state_fidelity_ball}
To start with, we note the following two facts: (1) Within $S_{\psi,f}$, the state which has the maximum geometric measure of entanglement can be chosen to be pure state, see Eq. (\ref{pure_max}). (2) From the Vidal's result \cite{vidal_prob}, we know that optimal probability of converting a pure two-qubit state into another two-qubit state is given by the ratio of their geometric measures of entanglement,  as also discussed previously. Hence, for the above case, the best strategy would be to start from a pure state with maximum possible geometric measure of entanglement in $S_{\psi,f_1}$ and end up in a state with the minimum possible geometric measure of entanglement in $S_{\rho,f_2}$.

Using the above facts, we obtain 
\begin{eqnarray}
&&\mathcal{P}_{f_1 , f_2}(\psi\rightarrow\rho) = \min \left\{\frac{\max_{\psi'\in S_{\psi,f_1}}G(\psi')}{\min_{\rho'\in S_{\rho,f_2}}G(\rho')},1\right\} \label{prob2} \\
&=&\min\left\{\frac{\sin^{2}\left(\min \left\{\sin^{-1}\sqrt{G(\psi)} + \cos^{-1}\sqrt{f_1},\frac{\pi}{4}\right)\right\}}{\sin^{2}\left(\max\left\{\sin^{-1}\sqrt{G(\rho)} - \cos^{-1}\sqrt{f_2},0\right\}\right)},1\right\}, \nonumber 
\end{eqnarray}
where in the last step we use Eqs. (\ref{min_geo}) and (\ref{pure_max}). Defining 
\begin{align}
    m_2=&\sin^{-1}\sqrt{G(\psi)}-\sin^{-1}\sqrt{G(\rho)} + \cos^{-1}\sqrt{f_1} \nonumber\\
    &+\cos^{-1}\sqrt{f_2}, \nonumber
\end{align}
let us consider the case when $m_2\geq0$, which can be expressed as
\begin{eqnarray}
&&\sin^{-1}\sqrt{G(\psi)} + \cos^{-1}\sqrt{f_1}  \geq\\&& \quad
\sin^{-1}\sqrt{G(\rho)}  - \cos^{-1}\sqrt{f_2}.\nonumber
\end{eqnarray}
This implies, 
\begin{eqnarray}
&&\min \left\{\sin^{-1}\sqrt{G(\psi)} + \cos^{-1}\sqrt{f_1},\frac{\pi}{4}\right\} \geq\nonumber\\ &&\quad\max\left\{\sin^{-1}\sqrt{G(\rho)} - \cos^{-1}\sqrt{f_2},0\right\}.
\end{eqnarray}
Therefore, from Eq. (\ref{prob2}),
\begin{equation}
    \mathcal{P}_{f_1 , f_2}(\psi\rightarrow\rho) =1.
\end{equation}
In the other case when,
\begin{eqnarray}
&&\sin^{-1}\sqrt{G(\rho)} - \cos^{-1}\sqrt{f_2}>\nonumber\\
&&\sin^{-1}\sqrt{G(\psi)} + \cos^{-1}\sqrt{f_1} > 0,
\end{eqnarray}
we have
\begin{equation}
\mathcal{P}_{f_1 , f_2}(\psi\rightarrow\rho) =\frac{\sin^{2}\left(\sin^{-1}\sqrt{G(\psi)} + \cos^{-1}\sqrt{f_1}\right)}{\sin^{2}\left(\sin^{-1}\sqrt{G(\rho)} - \cos^{-1}\sqrt{f_2}\right)}.
\end{equation}
Combining both these cases, we arrive at our final result.
\begin{eqnarray}
&&\mathcal{P}_{f_1,f_2}(\psi\rightarrow\rho) = \\&&\begin{cases}
1\,\,\,\mathrm{for}\,\,\,m_2 \geq 0,\\
\frac{\sin^{2}\left(\sin^{-1}\sqrt{G(\psi)} + \cos^{-1}\sqrt{f_1}\right)}{\sin^{2}\left(\sin^{-1}\sqrt{G(\rho)} - \cos^{-1}\sqrt{f_2}\right)}\,\,\,\mathrm{otherwise},
\end{cases}\nonumber
\end{eqnarray}
where, $m_2=\sin^{-1}\sqrt{G(\psi)}-\sin^{-1}\sqrt{G(\rho)} + \cos^{-1}\sqrt{f_1} +\cos^{-1}\sqrt{f_2}$.

\section{Optimal fidelity of the PR box with the set of local correlations}\label{app_prbox_ditill}
A non-signaling polytope corresponding to two inputs ($x,y\in\{0,1\}$) and two outputs ($a,b\in\{0,1\}$) has 24 vertices \cite{BarrettPhysRevA.71.022101}. Sixteen of them correspond to local deterministic boxes and are described by the following probability distribution \cite{BarrettPhysRevA.71.022101}
\begin{equation}
    P^{\alpha \beta \gamma \delta}(ab|xy)=
    \begin{cases}
    1 \,\, \mbox{for}\,\, a=\alpha x\oplus\beta, b=\gamma y\oplus\delta \\
    0 \,\, \mbox{otherwise},
    \end{cases}
\end{equation}
where $\alpha,\beta,\gamma,\delta\in\{0,1\}$. Here $\oplus$ denotes addition modulo 2. Convex combination of these 16 boxes will give all the local boxes in the local polytope. The remaining 8 boxes are known as PR boxes and represented as \cite{PRbox,BarrettPhysRevA.71.022101}
\begin{equation}
    P^{\alpha \beta \gamma}(ab|xy)=
    \begin{cases}
    \frac{1}{2} \,\, \mbox{for}\,\, a\oplus b=xy\oplus\alpha x\oplus\beta y\oplus\gamma \\
    0 \,\, \mbox{otherwise},
    \end{cases}
\end{equation}
where $\alpha,\beta,\gamma\in\{0,1\}$. Note that $\sum_{a,b} P(ab|xy)=1$ for all $x$ and $y$. For simplicity, we represent a box by a vector of sixteen elements. As an example, we consider a PR box with $\alpha=\beta=\gamma=0$ as following
\begin{equation}
\mathcal{B}_{\textrm{pr}}=\left(\frac{1}{2},0,0,\frac{1}{2},\frac{1}{2},0,0,\frac{1}{2},\frac{1}{2},0,0,\frac{1}{2},0,\frac{1}{2},\frac{1}{2},0\right)^T.
\end{equation}
Here, the ordering of $ab|xy$ is in the following order\\
$00|00,01|00,10|00,11|00,\\00|01,01|01,10|01,11|01,\\00|10,01|10,10|10,11|10,\\00|11,01|11,10|11,11|11$.\\
Consider the following isotropic box
\begin{equation}
\mathcal{B}_{\eta} = \eta\mathcal{B}_{\textrm{pr}}+(1-\eta)\mathcal{B}_{\textrm{random}},
\end{equation}
where 
\begin{eqnarray}
&&\mathcal{B}_{\textrm{random}}=\\
&&\left(\frac{1}{4},\frac{1}{4},\frac{1}{4},\frac{1}{4},\frac{1}{4},\frac{1}{4},\frac{1}{4},\frac{1}{4},\frac{1}{4},\frac{1}{4},\frac{1}{4},\frac{1}{4},\frac{1}{4},\frac{1}{4},\frac{1}{4},\frac{1}{4}\right)^T.\nonumber
\end{eqnarray}
For $\eta=1/2$, $\mathcal{B}_{\frac{1}{2}}$ becomes a local box \cite{Wolfe_2020}.
Since the ``boxes'' are classical-classical channels, (min-)fidelity between any two boxes (say $\mathcal{B}_{1}$ and $\mathcal{B}_2$) can be defined as 
\begin{align}\nonumber
    F(\mathcal{B}_1,\mathcal{B}_2)=\min_{q(x,y)} F\Big(&\sum_{x,y}q(x,y)p_{1}(ab|xy),\nonumber\\
    &\sum_{x,y}q(x,y)p_{2}(ab|xy)\Big).
\end{align}
Here, $q(x,y)$ is the input probability distribution, and $p_1(ab|xy)$ and $p_2(ab|xy)$ are the conditional probability distributions for getting the outputs $a,b$ (given inputs $x,y$) of boxes $\mathcal{B}_1$ and $\mathcal{B}_2$ respectively. From the concavity of the root fidelity, it follows that 
\begin{eqnarray}\nonumber
 &&\sqrt{F(\sum_{x,y}q(x,y)p_{1}(ab|xy),\sum_{x,y}q(x,y)p_{2}(ab|xy))} \\ &\geq&\sum_{x,y}q(x,y)\sqrt{F(p_{1}(ab|xy),p_{2}(ab|xy))} \nonumber \\
 &\geq&\min_{x,y}\sqrt{F(p_1(ab|xy),p_2(ab|xy))}.
\end{eqnarray}
This implies,
\begin{align}\nonumber
    F(\mathcal{B}_1,\mathcal{B}_2)&=\min_{q(x,y)} F\Big(\sum_{x,y}q(x,y)p_{1}(ab|xy),\nonumber\\
    &\quad\quad\quad\quad\quad\sum_{x,y}q(x,y)p_{2}(ab|xy)\Big)\nonumber\\
    &=\min_{x,y}F(p_1(ab|xy),p_2(ab|xy)).
\end{align}
Therefore, it is enough to consider the minimization over the inputs. Using this, the fidelity between $\mathcal{B}_{\textrm{pr}}$ and $\mathcal{B}_{\eta}$ can be calculated as
\begin{equation}
  F(\mathcal{B}_{\textrm{pr}},\mathcal{B}_{\eta})=\frac{\eta+1}{2}.
\end{equation}
This implies $F(\mathcal{B}_{\textrm{pr}},\mathcal{B}_{\frac{1}{2}})=\frac{3}{4}$.
 Now we will prove that there does not exists any local box which has fidelity greater than $3/4$ with $\mathcal{B}_{\textrm{pr}}$.

Let's take a general local box given by $\mathcal{B}_{\textrm{local}}=p(ab|xy)$. One can write joint probabilities in terms of marginals ($\langle A_x\rangle$) and correlations ($\langle A_xB_y\rangle$) \cite{Brunner_2014}
\begin{align}\label{correlators}
p(ab|xy)=\frac{1}{4}\Big(&1+(-1)^{a}\langle A_x\rangle+(-1)^{b}\langle B_y\rangle+\nonumber\\&(-1)^{a+b}\langle A_xB_y\rangle\Big),
\end{align}
where $\langle A_x\rangle=\sum_a (-1)^a p(a|x)$, $\langle B_y\rangle=\sum_b (-1)^b p(b|y)$, $\langle A_x B_y\rangle=\sum_{a,b} (-1)^{a+b} p(ab|xy)$, $p(a|x)=\sum_b p(ab|xy)$ and $p(b|y)=\sum_a p(ab|xy)$. Note that all local boxes satisfy the following Bell inequalities \cite{Brunner_2014}
\begin{equation}\label{bell_ineq}
2\geq \langle A_0B_0\rangle +\langle A_0B_1\rangle +\langle A_1B_0\rangle -\langle A_1B_1\rangle\geq-2.
\end{equation}
Permuting 0,1, we obtain other inequalities. Therefore there are 8 Bell inequalities.
The fidelity between this general local box $\mathcal{B}_{\textrm{local}}$ and $\mathcal{B}_{\textrm{pr}}$ is given by
\begin{align}
 &F(\mathcal{B}_{\textrm{pr}},\mathcal{B}_{\textrm{local}})=\nonumber\\
 &\quad\min\Bigg\{\left(\sqrt{\frac{p(00|00)}{2}}+\sqrt{\frac{p(11|00)}{2}}\right)^{2}
 ,\nonumber\\
 &\quad\quad\quad\quad\left(\sqrt{\frac{p(00|01)}{2}}+\sqrt{\frac{p(11|01)}{2}}\right)^{2},\nonumber\\
 &\quad\quad\quad\quad\left(\sqrt{\frac{p(00|10)}{2}}+\sqrt{\frac{p(11|10)}{2}}\right)^{2},\nonumber\\
 &\quad\quad\quad\quad\left(\sqrt{\frac{p(01|11)}{2}}+\sqrt{\frac{p(10|11)}{2}}\right)^{2}\Bigg\}.
\end{align}
If $F(\mathcal{B}_{\textrm{pr}},\mathcal{B}_{\textrm{local}})>\frac{3}{4}$, then following inequalities must hold
\begin{align}
   & \left(\sqrt{\frac{p(00|xy)}{2}}+\sqrt{\frac{p(11|xy)}{2}}\right)^{2}\nonumber\\
   &\quad>\frac{3}{4}\,\,\textrm{for}\,\,(x,y)\neq (1,1) \nonumber \\
  & \textrm{and}\,\, \left(\sqrt{\frac{p(01|11)}{2}}+\sqrt{\frac{p(10|11)}{2}}\right)^{2}>\frac{3}{4}.
\end{align}
Simplifying the above inequalities, we get
\begin{align}\nonumber
  & 2 p(00|xy)+2 p(11|xy)+4\sqrt{p(00|xy)p(11|xy)}\nonumber\\
  &\quad>3\,\,\textrm{for}\,\,(x,y)\neq (1,1) \nonumber\\
 & \textrm{and}\,\, 2 p(01|11)+ 2 p(10|11)+4\sqrt{p(01|11)p(10|11)}\nonumber\\
 &\quad\quad>3.
\end{align}
Using the above inequalities, along with positivity constraints ($p(ab|xy)\geq 0 $) and substituting Eq. (\ref{correlators}), we obtain the following inequalities
\begin{equation}
    \langle A_xB_y\rangle >\frac{3+(\langle A_x\rangle+\langle B_y\rangle)^{2}}{6}\,\,\textrm{for}\,\,(x,y)\neq (1,1)
\end{equation}
and
\begin{equation}
   - \langle A_1B_1\rangle >\frac{3+(\langle A_1\rangle-\langle B_1\rangle)^{2}}{6}.
\end{equation}
Summing these inequalities, we arrive at
\begin{eqnarray}\label{violation}
   && \langle A_0B_0\rangle +\langle A_0B_1\rangle +\langle A_1B_0\rangle -\langle A_1B_1\rangle >\\
   &&\frac{12+\sum_{(x,y)\neq (1,1)}(\langle A_x\rangle+\langle B_y\rangle)^{2}+(\langle A_1\rangle-\langle B_1\rangle)^{2}}{6}.\nonumber
\end{eqnarray}
Note that, $\sum_{(x,y)\neq (1,1)}(\langle A_x\rangle+\langle B_y\rangle)^{2}$ and $(\langle A_1\rangle-\langle B_1\rangle)^{2}$ are non-negative. Therefore, Eq. (\ref{violation}) is a contradiction to Eq. (\ref{bell_ineq}). This proves that the maximum fidelity of $\mathcal{B}_{\textrm{pr}}$ with the set of local boxes is $\frac{3}{4}$.

\bibliographystyle{quantum}
\bibliography{imaginarity}

\end{document}